\documentclass[conference]{IEEEtran}
\IEEEoverridecommandlockouts
\def\BibTeX{{\rm B\kern-.05em{\sc i\kern-.025em b}\kern-.08em
		T\kern-.1667em\lower.7ex\hbox{E}\kern-.125emX}}

\usepackage{booktabs}
\usepackage{times}
\usepackage{color}
\usepackage{balance}
\usepackage{url}

\usepackage{tabularx}
\usepackage{siunitx} 
\usepackage{epstopdf} 
\usepackage{epsfig,tabularx,subfigure,multirow, graphicx}
\usepackage{enumitem}
\usepackage{caption}
\usepackage{hyperref}
\usepackage{amsthm,amsmath,amssymb,amsfonts}

\newcolumntype{L}[1]{>{\raggedright\arraybackslash}p{#1}}
\newcolumntype{C}[1]{>{\centering\arraybackslash}p{#1}}
\newcolumntype{R}[1]{>{\raggedleft\arraybackslash}p{#1}}

\usepackage[linesnumbered,ruled,vlined]{algorithm2e}
\SetKwRepeat{Do}{do}{while}
 
\SetCommentSty{mycommfont}

\long\def\comment#1{}

\newcommand{\nop}[1]{}

\newtheorem{theorem}{\bf Theorem}[section]

\newtheorem{example}{\bf Example}

\theoremstyle{remark}

\theoremstyle{definition}
\newtheorem{definition}{\bf Definition}

\newcommand{\revision}[1]{\textcolor{black}{#1}}

\newcommand{\entity}[1]{\mathcal{#1}}
\newcommand{\algvar}[1]{\mathcal{#1}}

\newcommand{\vectorfont}[1]{\boldsymbol{#1}}

\newcommand{\problemDefineSimpleName}{PSDEP}
\newcommand{\problemDefineTotalName}{Private Spatial Distribution Estimation Problem}

\newcommand{\solutionGeneralName}{SAM}
\newcommand{\solutionGeneralTotalName}{Spatial Area Mechanism}
\newcommand{\solutionD}{HUEM}
\newcommand{\solutionDTotalName}{Hybrid Uniform-Exponential Mechanism}
\newcommand{\solutionB}{DAM}
\newcommand{\solutionBTotalName}{Disk Area Mechanism}
\newcommand{\solutionBNS}{DAM-NS}
\newcommand{\solutionBNSTotalName}{Disk Area Mechanism with Non-Shrink}

\newcommand{\solutionCMPm}{MDSW}
\newcommand{\solutionCMPmTotalName}{Multi-dimensional Square Wave Mechanism}

\newcommand{\solutionCMPg}{SEM-Geo-I}
\newcommand{\solutionCMPgTotalName}{Subset Exponential Mechanism with Geo-I}

\newcommand{\solutionCMPTrA}{LDPTrace}
\newcommand{\solutionCMPTrATotalName}{Locally Differentially Private Trajectory Synthesis}

\newcommand{\solutionCMPTrB}{PivotTrace}
\newcommand{\solutionCMPTrBTotalName}{Trajecotry Data Collectin with Local Differential Privacy}

\begin{document}

\title{Numerical Estimation of Spatial Distributions under Differential Privacy}

\author{
	{Leilei Du{\small$~^{1}$}, Peng Cheng{\small$~^{1}$}, Libin Zheng{\small$~^{2}$}, Xiang Lian{\small$~^{3}$}, Lei Chen{\small$~^{4, 5}$}, Wei Xi{\small$~^{6}$}, Wangze Ni{\small$~^{7*}$}}\\
	\fontsize{10}{10}\itshape
	$~^{1}$East China Normal University, China; $~^{2}$Sun Yat-sen University, China; $~^{3}$Kent State University, USA; \\
	$~^{4}$HKUST (GZ), China;
	$~^{5}$HKUST, China; $~^{6}$Xi'an Jiaotong University, China;
	$~^{7}$Zhejiang University, China\\
	\fontsize{9}{9}\upshape
	leileidu@stu.ecnu.edu.cn; pcheng@sei.ecnu.edu.cn; zhenglb6@mail.sysu.edu.cn; xlian@kent.edu; leichen@cse.ust.hk;\\ xiwei@xjtu.edu.cn; niwangze@zju.edu.cn
	\thanks{*Wangze Ni is also with The State Key Laboratory of Blockchain and Data Security; Hangzhou High-Tech Zone (Binjiang) Institute of Blockchain and Data Security.}
}

\maketitle

\begin{abstract}
	Estimating spatial distributions is important in data analysis, such as traffic flow forecasting and epidemic prevention. 
	To achieve accurate spatial distribution estimation, the analysis needs to collect sufficient user data. 
	However, collecting data directly from individuals could compromise their privacy. Most previous works focused on private distribution estimation for one-dimensional data, which does not consider spatial data relation and leads to poor accuracy for spatial distribution estimation. In this paper, we address the problem of private spatial distribution estimation, where we collect spatial data from individuals and aim to minimize the distance between the actual distribution and estimated one under Local Differential Privacy (LDP).
	To leverage the numerical nature of the domain, we project spatial data and its relationships onto a one-dimensional distribution. We then use this projection to estimate the overall spatial distribution.
	Specifically, we propose a reporting mechanism called Disk Area Mechanism (DAM), which projects the spatial domain onto a line and optimizes the estimation using the sliced Wasserstein distance. Through extensive experiments, we show the effectiveness of our DAM approach on both real and synthetic data sets, compared with the state-of-the-art methods, such as Multi-dimensional Square Wave Mechanism (MDSW) and Subset Exponential Mechanism with Geo-I (SEM-Geo-I). Our results show that our DAM always performs better than MDSW and is better than SEM-Geo-I when the data granularity is fine enough.
\end{abstract}

\section{Introduction}
With the popularity of smart devices and the high quality of wireless networks, people can easily access the Internet and communicate with online services. Convenient online service platforms, such as ride-hailing apps, collect user data, analyze it, and provide better services in return.
For example, collecting vehicle locations and analyzing traffic flow can help ride-hailing drivers avoid traffic jams.
However, directly collecting data could compromise individuals' privacy, leading to users refusing to share their information.
In traffic flow forecasting, if a driver submits his/her locations to the platform over a period (e.g., a month), a malicious platform attacker could predict the driver's activity range and surveil him/her.

Differential Privacy (DP) \cite{DBLP:conf/icalp/Dwork06} is a privacy standard that resolves conflicts between data privacy and data analysis with the aid of a trusted server. To further avoid information leakage on the trust server, a method for DP in a local setting, called Local Differential Privacy (LDP), has been proposed recently \cite{DBLP:conf/stoc/BassilyS15}.
In LDP, there are multiple \emph{users} and one \emph{analyst}.
First, all users randomize their actual information (to protect their privacy) themselves and send it to the analyst.
Then, the analyst estimates the users' distribution based on the randomized information.
These two steps by the users and analyst are encapsulated as Frequency Oracle \cite{DBLP:conf/uss/WangBLJ17} (FO) for count queries under LDP.
As an efficient tool, FO has been widely used to resolve distribution estimation under LDP~\cite{DBLP:conf/uss/WangBLJ17,DBLP:conf/icde/WangXYHSS018_Minimax,DBLP:conf/icde/WangXYHSS018_PrivTrie,DBLP:conf/sigmod/Li0LLS20,DBLP:journals/pvldb/CormodeMM21} and various private query issues, such as private range queries~\cite{DBLP:conf/ccs/Du0BLJ0021,DBLP:conf/sigmod/WangDZHHLJ19,DBLP:journals/pvldb/Yang0L0S20,DBLP:conf/ccs/Du0BLJ0021,DBLP:conf/icde/WangW0N0TG023}.

Traditional FO performs well for estimating categorical data distribution (i.e., data without order)~\cite{DBLP:conf/sigmod/Li0LLS20}.
Similar to existing studies \cite{DBLP:conf/uss/WangBLJ17,DBLP:journals/pvldb/CormodeMM21,DBLP:conf/infocom/WangNWXYH17}, we can directly dividing the spatial area into unrelated grids and apply traditional FO to estimate the spatial distribution.
However, in the spatial (2-Dim) context, there is a strong ordinal relationship between any pair of data points.
For example, a heavily congested traffic junction is more likely to cause blockages at nearby junctions than those farther away.
Similarly, a COVID-19 affected area is more likely to lead to outbreaks in surrounding areas than in distant ones.
Therefore, if we use traditional FO, the ordinal relationship in the  domain may be ignored, which may lead to poor estimation accuracy~\cite{DBLP:conf/sigmod/Li0LLS20}.

\begin{figure}[t!]
	\centering
	\includegraphics[width=0.5\linewidth]{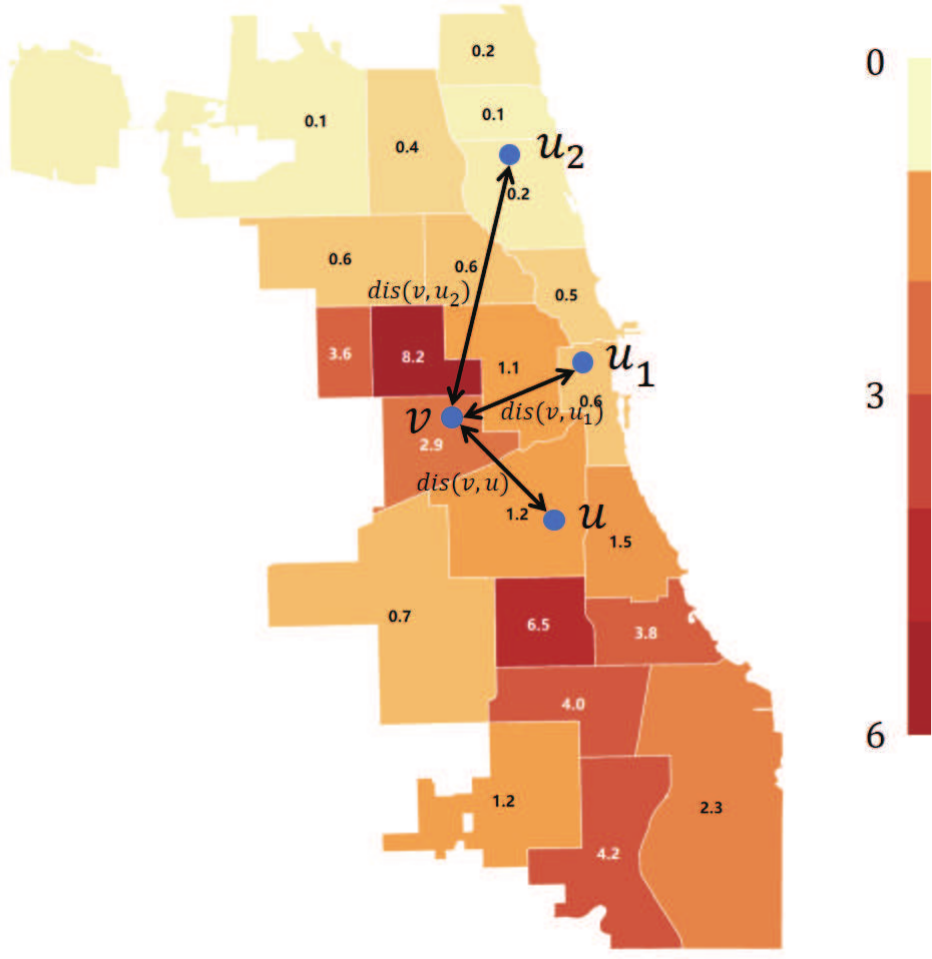}
	\caption{shooting victims per 1,000 residents of Chicago in 2021}\label{fig:Introduction_A}
\end{figure}

In this paper, we study the problem of private spatial distribution estimation, in which each user publishes their data under LDP, and the analyst estimates the distribution of these users to minimize the difference between the actual and recovered density distributions.
To illustrate the motivation for this problem, we consider the following example:\vspace{-2ex}
\begin{example}\label{example_1}
	As shown in Figure~\ref{fig:Introduction_A}, an analyst wants to estimate the shooting distribution in Chicago.
	To do this, the analyst needs to collect shooting points in this area and count the number of points at each location.
	In order to protect the actual shooting locations and maintain social stability, the police use an LDP mechanism to randomize the locations and send these randomized locations to the analyst.
	After collecting all randomized locations, the analyst can estimate the shooting distribution of Chicago.
	
	Take point $v$ as a shooting instance.
	Given another shooting instance point $u$, we denote the distance between $v$ and $u$ as $dis(v,u)$.
	If we regard locations as categorical points, according to traditional FO (i.e., Categorical Frequency Oracle (CFO)~\cite{DBLP:conf/sigmod/Li0LLS20}), $v$ will be published as $u_1$ or $u_2$ with the same low probability.
	This method neglects the location ordinal relationship among the nodes (i.e., $dis(v,u_1)<dis(v,u_2)$).
	In this way, it is both dangerous for the citizens of Chicago (the shooter can easily move to $u_1$ rather than $u_2$) and leads to poor estimation accuracy for shooting distribution.
	Considering the location ordinal relationship, $u_1$ is closer to $v$ than $u_2$.
	Thus, for $v$'s randomized points, the probability of choosing $u_1$ should be higher than that of $u_2$.
\end{example}

Many studies use mean absolute error, variances or Kullback-Leibler (KL) Divergence~\cite{ash2012information} as metrics for distribution estimation~\cite{DBLP:conf/sigmod/Li0LLS20}.
However, none of these metrics effectively capture the ordinal relationship.
In this paper, we use the common-used 2-Dim Wasserstein distance~\cite{panaretos2019statistical} to measure the difference between two distributions while capturing the spatial ordinal relationship.
Based on 2-Dim Wasserstein distance, we propose a general spatial distribution estimation mechanism called \emph{\solutionGeneralTotalName{}} (\solutionGeneralName{}), which achieves $\epsilon$-LDP.
Additionally, we introduce a simple implementation mechanism of \solutionGeneralTotalName{} called \emph{\solutionDTotalName{}} (\solutionD{}).
To optimize \solutionGeneralName{}, we need the close forms of 2-Dim Wasserstein Distance.
However, except for the 2-Dim normal (Gaussian) distribution, there are no closed forms for 2-Dim Wasserstein distance when the data dimension is greater than 1~\cite{panaretos2019statistical}, making it difficult to optimize the estimation utility.
A direct method is to optimize each dimension estimation according to the 1-Dim Wasserstein distance and then combine them (MSW~\cite{DBLP:journals/pvldb/Yang0L0S20}).
However, this approach loses the spatial ordinal relationship.
To address this problem, we use the Radon Transform~\cite{helgason1980radon} to project spatial data onto one-dimensional (1-Dim) data and transform the 2-Dim Wasserstein distance into the Sliced Wasserstein distance~\cite{DBLP:conf/nips/KolouriNSBR19}.
Based on the sliced Wasserstein distance metric, we propose \emph{\solutionBTotalName{}} (\solutionB{}) and prove that it is optimal among all types of \solutionGeneralName{}.
To implement our \solutionB{} on real data (discretized domain), we design a \emph{grid partition and shrinkage} method to effectively bucketize the data.
Our \solutionB{} achieves a lower 2-Dim Wasserstein distance (between the recovered and actual density distributions) than the state-of-the-art \solutionCMPmTotalName{} (\solutionCMPm{})~\cite{DBLP:journals/pvldb/Yang0L0S20} and the categorical \solutionCMPgTotalName{} (\solutionCMPg{})~\cite{DBLP:conf/infocom/WangNWXYH17}.
The contributions of this paper are as follows:
\begin{itemize}[leftmargin=*]
	\item We formally define our \problemDefineTotalName{} (\problemDefineSimpleName{}) in Section~\ref{ProblemDefinition} and propose a mechanism structure called \solutionGeneralTotalName{} (\solutionGeneralName{}) and a direct baseline method called \solutionDTotalName{} (\solutionD{}) in Section~\ref{DirectApproach}.
	\item We propose \solutionBTotalName{} (\solutionB{}) and analyze how to choose the best norm distance $b$ in Section~\ref{BestMethod}.
	\item We introduce the implementation of our \solutionB{} including grid partitioning, shrinkage and post-process in Section~\ref{Implement}.
	\item We unify local differential privacy mechanisms (e.g., \solutionB{}) and Geo-I mechanisms (e.g., \solutionCMPg{} \cite{DBLP:conf/infocom/WangNWXYH17}) by the \emph{local privacy mechanism} \cite{DBLP:conf/ccs/ShokriTTHB12} and conduct experimental evaluations of our proposed method on both real and synthetic datasets to demonstrate its efficiency and effectiveness in Section~\ref{Experiment}.
\end{itemize}

\section{Related Work}

\begin{table}[t!]
	\caption{The summary of related studies.}
	\label{relatedwork_table}\vspace{1ex}
	\centering
		\resizebox{8.8cm}{!}{
			\begin{tabular}{|c|c|c|c|c|}
				\hline
				\textbf{classified name}   & \textbf{mechanism name}       & \textbf{catch numeric}    & \textbf{meet SDP}    &\textbf{locally} \\ \hline
				central privacy        & PSD+Geocast~\cite{DBLP:journals/tmc/ToGFS17}  & $\times$    & $\surd$           & $\times$     \\ \hline
				\multirow{2}{*}{local privacy}  & ID-LDP~\cite{DBLP:conf/icde/Gu0XC20}    & $\times$   & $\surd$         & $\surd$     \\ \cline{2-5}
				&{Geo-I}~\cite{DBLP:conf/ccs/AndresBCP13}        & 2-Dim         & $\surd$           & $\surd$    \\ \hline
				categoric      & Bucket+CFO~\cite{DBLP:conf/uss/WangBLJ17,DBLP:journals/pvldb/CormodeMM21}              & $\times$               & $\surd$           & $\surd$          \\ \hline
				\multirow{4}{*}{numeric} & {SEM-Geo-I}~\cite{DBLP:conf/infocom/WangNWXYH17}               & 2-Dim                & $\surd$           & $\surd$          \\ \cline{2-5}
				& SR~\cite{DBLP:conf/icde/WangXYHSS018_Minimax}                      & 1-Dim                & $\times$          & $\surd$          \\ \cline{2-5}
				& PM~\cite{DBLP:conf/icde/WangXYHSS018_PrivTrie}                      & 1-Dim                & $\times$          & $\surd$          \\ \cline{2-5}
				& \revision{\multirow{2}{*}{SW-EMS~\cite{DBLP:conf/sigmod/Li0LLS20}}}                  & \revision{\multirow{2}{*}{1-Dim}}                & \revision{\multirow{2}{*}{$\times$}}          & \revision{\multirow{2}{*}{$\surd$}}          \\ \cline{1-1}
				one-dimensional                   &                   &                 &          &           \\ \hline
				\multirow{5}{*}{multi-dimensional} & PSD~\cite{DBLP:conf/icde/CormodePSSY12}      & $\times$               & $\surd$           & $\times$          \\ \cline{2-5}
				& AG~\cite{DBLP:conf/sigmod/QardajiYL14}        & $\times$               & $\surd$           & $\times$          \\ \cline{2-5}
				& HIO~\cite{DBLP:conf/sigmod/WangDZHHLJ19}        & $\times$               & $\surd$           & $\surd$          \\ \cline{2-5}
				& AHEAD~\cite{DBLP:conf/ccs/Du0BLJ0021}        & $\times$               & $\surd$           & $\surd$          \\ \cline{2-5}
				& {MSW, HDG}~\cite{DBLP:journals/pvldb/Yang0L0S20}                 & 1-Dim                & $\surd$           & $\surd$          \\ \hline
				$\backslash$                     &{\textbf{Our mechanism}}  & 2-Dim                & $\surd$           & $\surd$          \\ \hline
			\end{tabular}
		}
\end{table}

The privacy protection is an important issue in spatial data statistics. It requires obtaining accurate statistical results while protecting individuals' information from being released. Differential privacy~\cite{DBLP:conf/icalp/Dwork06} is a key tool for privacy protection and privacy-preserving data release. We classify related work based on differential privacy into three dimensions and summarize it in Table~\ref{relatedwork_table}.

\noindent\textbf{{Central / Local Differential Privacy.}}
Conventional differential privacy requires a trusted third party to collect individuals' data and randomize it under differential private mechanisms~\cite{DBLP:journals/tmc/ToGFS17}. However, the third party may be attacked by malicious entities, hindering individuals from sharing their information. To address this issue, local differential privacy (LDP)~\cite{DBLP:conf/focs/KasiviswanathanLNRS08,DBLP:conf/focs/DuchiJW13} is proposed, where individuals randomize their own information and then report the randomized messages to the estimator. Gu et al.~\cite{DBLP:conf/icde/Gu0XC20} propose input-discriminative LDP, which can satisfy different privacy levels required by different individuals simultaneously. Andr{\'{e}}s and Bordenabe~\cite{DBLP:conf/ccs/AndresBCP13} propose Geo-Indistinguishability (Geo-I), which provides high privacy within short distances and low privacy with long distances. However, both of these designs detriment the privacy of LDP.

\noindent\textbf{{Categorical / Numerical Frequency Oracle.}}
To handle the issue of releasing numeric data under Local Differential Privacy (LDP), a popular method is to apply Categorical Frequency Oracle (CFO, FO)~\cite{DBLP:conf/uss/WangBLJ17,DBLP:journals/pvldb/CormodeMM21}. The basic process is to first divide the numeric data into several buckets and then use CFO to estimate the result. However, simply dividing the data leads to information loss during comparison.
Duchi et al.~\cite{DBLP:conf/icde/WangXYHSS018_Minimax} propose Stochastic Rounding (SR) to handle numerical settings. In SR, a value $v$ in the interval $[-1,1]$ returns $-1$ with probability $\frac{1}{2}-\frac{e^\epsilon-1}{2(e^\epsilon+1)}v$ and $1$ with probability $\frac{1}{2}+\frac{e^\epsilon-1}{2(e^\epsilon+1)}v$.
Wang et al.~\cite{DBLP:conf/icde/WangXYHSS018_PrivTrie} propose the Piecewise mechanism (PM), where the input domain is $[-1,1]$ and the output domain is $[-s,s]$, with $s=\frac{e^{\epsilon/2}+1}{e^{\epsilon/2}-1}$. Given an input point $v$, it returns a point in the subinterval $[\frac{e^{\epsilon/2}v-1}{e^{\epsilon/2}-1},\frac{e^{\epsilon/2}v+1}{e^{\epsilon/2}-1}]$ with probability $\frac{e^{\epsilon/2}(e^{\epsilon/2}-1)}{2(e^{\epsilon/2}+1)}$, and the complement subinterval with probability $\frac{e^{\epsilon/2}-1}{2e^{\epsilon/2}(e^{\epsilon/2}+1)}$.
Note that both SR and PM focus on the specific task of mean estimation.
Li et al.~\cite{DBLP:conf/sigmod/Li0LLS20} propose the Square Wave mechanism with Expectation Maximization Smoothing (SW-EMS) to handle numerical distribution under local differential privacy (LDP).
SW-EMS is a new numeric frequency oracle that makes full use of ordinal relations to obtain much more accurate estimations without breaching privacy. However, SW-EMS only focuses on one-dimensional data, and is therefore not suitable for estimating spatial distributions (SDP).
Wang et al.~\cite{DBLP:conf/infocom/WangNWXYH17} propose the Subset Exponential Mechanism under $\epsilon$-Geo-I constraints (SEM-Geo-I). SEM-Geo-I can achieve accurate estimation, however, it only provides strictly weaker privacy based on Geo-I.

\noindent\textbf{{One / Multiple Dimensional Data Estimation.}}
Several works have been proposed for handling spatial data with traditional differential privacy. Cormode et al.~\cite{DBLP:conf/icde/CormodePSSY12} design a new structure called PSD, which utilizes indexing methods such as quadtrees and kd-trees to generate spatial decompositions for describing the data distribution. Similarly, Qardaji et al.~\cite{DBLP:conf/sigmod/QardajiYL14} present an Adaptive Grid (AG) approach to release a synopsis for 2-Dim geospatial data. However, both PSD and AG require the aid of a trusted third party.
Yang et al.~\cite{DBLP:journals/pvldb/Yang0L0S20} propose the Multiplied Square Wave (MSW) mechanism, which extends the SW-EMS~\cite{DBLP:conf/sigmod/Li0LLS20} mechanism. MSW provides an accurate estimation for multi-dimensional data under Local Differential Privacy (LDP). However, it can only capture the correlation in one dimension, which leads to high error. 
In order to capture the correlation among different dimensions, they propose the Hybrid-Dimensional Grids (HDG) method.
HDG divides the $n$-Dim dimensional data into 1-Dim and 2-Dim grids and use these grids to capture the correlation among different dimensions in range query. However 1-Dim grids in HDG may still destroy the correlation among different dimension data.
Du et al.~\cite{DBLP:conf/ccs/Du0BLJ0021} propose the Adaptive Hierarchical Decomposition~(AHEAD) method based on HIO~\cite{DBLP:conf/sigmod/WangDZHHLJ19} to further improve the private range query by adaptively choosing the granularity of domain composition.
However HDG and AHEAD do not catch the numeric (the distance) in spatial relationship from different randomized points to the real points.

Our mechanism can not only catch the numeric relationship and accurately estimate spatial distribution estimation under LDP, but also combine with the methods of HIO, HDG and AHEAD to further improve the accuracy in private range query.

\section{Problem definition}\label{ProblemDefinition}

In this section, we provide basic notations and preliminaries, distance metrics, and formal definition of our \textit{\problemDefineTotalName{}} (\problemDefineSimpleName{}).
\revision{
	Table~\ref{tbl:notations} summarizes the key notations used throughout this paper.
}
\begin{table}[t!]
	\caption{Notations.}
	\label{tbl:notations}
	\centering
	\scalebox{1}{
		\revision{
			\begin{tabular}{c|c}
				\hline
				\textbf{Notations}        & \textbf{Description}                                       \\ \hline
				$k$-Dim             & $k$-dimension                                        \\ \hline
				$\entity{D}$, $\mathcal{I}$             & the input domain                                        \\ \hline
				$\tilde{\entity{D}}$, $\mathcal{T}$             & the output domain                                        \\ \hline
				$\mathcal{\hat{I}}$            & the inferred (estimated) domain of $\mathcal{I}$ domain                                        \\ \hline
				$D$             & an input instance                                         \\ \hline
				$\tilde{D}, O$             & an output instance instance                                         \\ \hline
				$\vectorfont{v}$             & a spatial data point                                        \\ \hline
				$\tilde{\vectorfont{v}}$             & a disturbed point of $\vectorfont{v}$                                        \\ \hline
				$M_{\vectorfont{v}}(\tilde{\vectorfont{v}})$            & the probability of randomizing $\vectorfont{v}$ as  $\tilde{\vectorfont{v}}$                                \\ \hline
				$W(\cdot)$             & the wave function                                 \\ \hline
				\multirow{2}{*}{$W_k^p$}             & a $k$ dimensional Wasserstein distance                                          \\
				& with $p$ norm cost function\\ \hline
				$W_k(\cdot)$             & the $p$-th root of $W_k^p$ (i.e., $W_k(\cdot)=\sqrt[p]{W_k^p}$ )                                 \\ \hline
				\multirow{2}{*}{$SW_k^p$}             & a $k$ dimensional sliced Wasserstein distance                                          \\
				& with $p$ norm cost function\\ \hline
				$b$             & a high dimensional radius                                        \\ \hline
				$L$             & the side length of an input instance                                        \\ \hline
				$g$             & the side length of a grid cell                                        \\ \hline
				$d$             & the number of cells along a side of grid length                                        \\ \hline
				$n$             & the number of cells in the grid                                        \\ \hline
			\end{tabular}
		}
	}
\end{table}

\subsection{Basic Notations and Preliminaries}
We use the notation $[a_1:a_2]$ to denote an integer series from $a_1$ to $a_2$ and abbreviate $[a_1:a_2]$ to $[a_2]$ when $a_1=1$.
We use $x\xleftarrow{\$}X$ to indicate uniformly sampling an element $x$ from set $X$.
We use $v$ to indicate a point with index $(x_v,y_v)$ in the Plan-Rectangular coordinate system~\cite{bossler2010coordinates} and $(r_v, \theta_v)$ in the Polar coordinate system~\cite{bossler2010coordinates}.
The input domain is denoted as $\mathcal{D}$, and the output domain is denoted as $\tilde{\mathcal{D}}$.
We use $\|M\|_p$ to denote the $p$-norm of any matrix $M$, where $p\in \mathbb{N}\cup\{\infty\}$.
We denote the inner product of matrices $A_1$ and $A_2$ as $A_1\cdot A_2$, the element-wise product (also called Hadamard product) as $A_1\bigodot A_2$.
For example, let $A_1=(c_{i,j})_{n\times n}$ and $A_2=(d_{i,j})_{n\times n}$ for $1\leq i,j\leq n$.
Then we have $A_1\bigodot A_2 = (c_{i,j}d_{i,j})_{n\times n}$ for $1\leq i,j\leq n$.
We abbreviate $k$-dimension as $k$-Dim.
When $k=2$, we also call the data as \emph{spatial} data.

We utilize Local Differential Privacy (LDP)~\cite{DBLP:conf/stoc/BassilyS15} to protect the privacy of original data locations.

\begin{definition}($\epsilon$-Local differential privacy, $\epsilon$-LDP~\cite{DBLP:conf/stoc/BassilyS15})\label{Def_LDP}.
	An algorithm $M(\cdot): D\to\tilde{D}$ satisfies $\epsilon$-local differential privacy ($\epsilon$-LDP), where $\epsilon\geq 0$ if and only if for any input values $v_1,v_2\in D$, we have
	$$\forall S\subset \tilde{D}: \Pr[M(v_1)\in S]\leq e^\epsilon \Pr[M(v_2)\in S],$$
	where $\tilde{D}$ denotes the set of all possible output of $M$.
\end{definition}
Based on the Local Differential Privacy model, a standard protocol called Frequency Oracle ($FO$)~\cite{DBLP:conf/ndss/0001LLSL20} for frequency estimation has been proposed. $FO$ is composed of two functions, namely, the \textit{randomized reporting function} $FO.T$ and the \textit{estimation function} $FO.E$. $FO.T$ is used to randomize the raw data into a kind of randomized data, while $FO.E$ is used to estimate the raw data based on the randomized data.

\subsection{Distance Metrics}
\begin{definition}(Wasserstein Distance~\cite{panaretos2019statistical}).
	Let $\mathcal{P}_p(\mathbb{R}^k)$ be the space of Borel probability measures on $\mathbb{R}^k$ with finite $p$-th moments, i.e. for all $\mu\in \mathcal{P}_p(\mathbb{R}^k)$, $\int_{\mathbb{R}^k}|x|^p<\infty$. Let $\mu_A, \mu_B\in \mathcal{P}_p(\mathbb{R}^k)$ then we define the $L_k^p$-Wasserstein distance as:
	{\scriptsize
		$$W_k^p(\mu_A,\mu_B)=\inf\left\{\int_{\mathbb{R}^k\times\mathbb{R}^k}|x-y|^pd\pi(x,y):\pi\in\Pi(\mu_A,\mu_B)\right\},$$
	}
	where $\inf$ is the infimum (greatest lower bound) function and $\Pi(\mu_A,\mu_B)$ is the complete set of joint distributions of $\mu_A$ and $\mu_B$.
\end{definition}

Wasserstein distance~\cite{panaretos2019statistical} (also called Earth Mover's distance) is a metric on probability distributions used to measure the minimal effort of probability mass from one distribution to another. It can be used to measure the similarity between two distributions.

\subsection{\problemDefineSimpleName{} Definition}
We give our problem definition in Definition~\ref{def:problem} as follows.
\begin{definition}(\problemDefineTotalName{}, \problemDefineSimpleName{}).\label{def:problem}
	Given a set of ordinal spatial values $V\subseteq \mathbb{R}^2$ with $\chi$ distinct values,
	a privacy budget $\epsilon$, a \problemDefineSimpleName{} is to design a frequency oracle mechanism $FO=<T,E>$ satisfying that for the actual distribution $D\in \mathbb{R}^{\chi}$ of $V$, $FO$ outputs $\tilde{D}\in \mathbb{R}^{\chi}$, where $FO.T$ satisfies $\epsilon$-LDP and the $L_2^2$-Wasserstein distance $W_2^2(D,\tilde{D})$ is minimized.
\end{definition}

\section{The \solutionDTotalName{}}\label{DirectApproach}
In this section, we first declare the definition of the input/output domain and the randomized function.
Then we propose a general \solutionGeneralTotalName{} (\solutionGeneralName{}) and prove that it satisfies $\epsilon$-LDP. 
After that we introduce a direct implementation mechanism of \solutionGeneralName{} called \solutionDTotalName{} (\solutionD{}), and analyze its accuracy.

\noindent
\textbf{Input/Output Domain.}
Without loss of generality, we define the input domain $\mathcal{D}=\{\vectorfont{v}|x_{\vectorfont{v}}\in[0,1],y_{\vectorfont{v}}\in[0,1]\}$ as a square with side length $1$. 
For any input point $\vectorfont{v}$, we define its \textit{$b$ distance set} as $DS_{b}(\vectorfont{v})=\{\vectorfont{u}|\|\vectorfont{u}-\vectorfont{v}\|_2\leq b\}$.
We define the output domain as the union set of all points' $DS_{b}$ in $\mathcal{D}$, namely,  $\tilde{\mathcal{D}}=\bigcup_{\vectorfont{v}\in\mathcal{D}}\{DS_{b}\}$. 
We call $b$ as the \textit{high probability radius}.
Figure~\ref{fig:Input and output domain} shows the input and output domains. The input domain $\mathcal{D}$ is the black square with side length $1$. The output domain $\tilde{\mathcal{D}}$ is the red rounded square. For a point $v\in\mathcal{D}$, its $DS_{b}$ is the point located within the green circle.

\noindent
\textbf{Randomized Reporting Function.} 
For any input point $v\in\mathcal{D}$, let $M_{\vectorfont{v}}: \tilde{\mathcal{D}}\to [0,1]$ be the probability density functions (PDF) over the output domain $\tilde{\mathcal{D}}$.
Spatially, $M_{\vectorfont{v}}(\tilde{\vectorfont{v}})$ means the probability of randomizing $\vectorfont{v}$ as $\tilde{\vectorfont{v}}$. 
We define the randomized reporting functions as a family of PDF over the output domain (i.e., $\{M_{\vectorfont{v}}(\cdot)\}_{\vectorfont{v}\in \mathcal{D}}$).

\begin{figure}[t!]\centering\vspace{2ex}
	\scalebox{0.32}[0.32]{\includegraphics{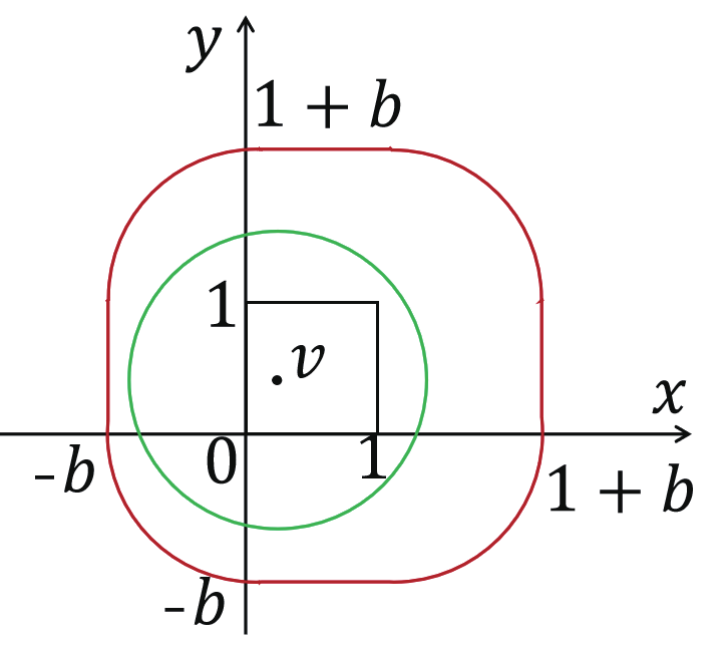}}
	\caption{\small Radon Transform and I/O domain with any real point.}
	\label{fig:Input and output domain}
\end{figure}
\noindent
\textbf{\solutionGeneralTotalName{} (\solutionGeneralName{}).} Based on the input/output domain and randomized reporting function, we propose our \solutionGeneralTotalName{} in Definition~\ref{def:SAM}. 
\begin{definition}(\solutionGeneralTotalName{} (\solutionGeneralName{})).\label{def:SAM}
	A randomized mechanism $\Psi: \mathcal{D}\to\tilde{\mathcal{D}}$ is an instance of Spatial Area Mechanism if for all $\vectorfont{v}\in\mathcal{D}$, there is a 2-dimension wave function $W:\mathbb{R}^2\to[q,e^\epsilon q]$ with constant $q>0$ and $\epsilon>0$ such that the output probability density function $M_{\vectorfont{v}}(\tilde{\vectorfont{v}})=W(\tilde{\vectorfont{v}}-\vectorfont{v})$ satisfies: 
	\begin{enumerate}[label=(\arabic*)]
		\item $W(\vectorfont{z})=q$ for $\parallel \vectorfont{z}\parallel_2>b$ \label{cond:w_1}
		\item $\iint_D W(\vectorfont{z})d\vectorfont{z}=1-(4b+1)q$ for $D=\{\vectorfont{z}| \parallel \vectorfont{z} \parallel_2 \leq b\}$ \label{cond:w_2}
	\end{enumerate}
\end{definition}
\revision{\solutionGeneralName{} is a general mechanism structure. 
	It declares a wave function $W$ with range between $q$ and $e^\epsilon q$.
	Based on the function $W$, it claims two conditions for the areas within and out of $b$ distance. 
	Noted that, in condition~\ref{cond:w_2} (distance within $b$), the distribution function of $W(\vectorfont{z})$ is not defined. 
	The only constrain is the integral in this area keeps $1-(4b+1)q$.}

\begin{theorem}\label{thrm:LDP}
	\solutionGeneralName{} satisfies $\epsilon$-LDP.
\end{theorem}
\begin{proof}
	For any two possible input values $\vectorfont{v}_1,\vectorfont{v}_2\in\mathcal{D}$ and any set of possible outputs $O\subseteq \tilde{\mathcal{D}}$ of \solutionGeneralName{}, we have
	\begin{equation}\label{AM1_Proof}
		{\scriptsize
			\begin{aligned}
				\frac{\textrm{Pr}[\solutionGeneralName{}(\vectorfont{v}_1)\in O]}{\textrm{Pr}[\solutionGeneralName{}(\vectorfont{v}_2)\in O]} &= \frac{\iint_{\tilde{\vectorfont{v}}\in O}M_{v_1}(\tilde{\vectorfont{v}})d\tilde{\vectorfont{v}}}{\iint_{\tilde{\vectorfont{v}}\in O}M_{v_2}(\tilde{\vectorfont{v}})d\tilde{\vectorfont{v}}} \\
				&\leq \frac{\iint_{\tilde{\vectorfont{v}}\in O}e^\epsilon q d\tilde{\vectorfont{v}}}{\iint_{\tilde{\vectorfont{v}}\in O}q d\tilde{\vectorfont{v}}} = e^\epsilon
			\end{aligned}
		}
	\end{equation}\vspace{-2ex}
\end{proof}

For any $\vectorfont{v}\in{\mathcal{D}}$ and $\tilde{\vectorfont{v}}\in\tilde{\mathcal{D}}$, 
it is reasonable to assume that the reporting probability $M_{\vectorfont{v}}(\tilde{\vectorfont{v}})$ decreases as the $dis(\vectorfont{v},\tilde{\vectorfont{v}})$ increases (similar to Reference~\cite{DBLP:conf/ccs/AndresBCP13}).
To model this relationship, we propose our \solutionDTotalName{} in Definition~\ref{def:HUEM}.

\begin{definition} (\solutionDTotalName{}, \solutionD{}).\label{def:HUEM}
	A \solutionGeneralName{} is called a \solutionDTotalName{} if the $W$ function satisfies: 
	\begin{equation}\label{direct_solution_equation}
		{\scriptsize
			\begin{aligned}
				W(\vectorfont{z})=
				\begin{cases}
					qe^{(1-\frac{\|\vectorfont{z}\|_2}{b})\epsilon},& \textrm{if}\;\;\|\vectorfont{z}\|_2\leq b   \\
					q,& otherwise
				\end{cases}
			\end{aligned}
		}
	\end{equation}
	where $q=\frac{\epsilon^2}{2\pi(e^\epsilon-1-\epsilon)b^2+4\epsilon^2 b+\epsilon^2}$.
\end{definition}
\solutionD{} is a type of \solutionGeneralName{} by adding constrain that $W$'s value increases exponentially with distance $\|\vectorfont{z}\|_2$ decreases.
From Equation~\eqref{direct_solution_equation}, we can see when $\vectorfont{z}=\vectorfont{0}$, $W(\vectorfont{z})=qe^\epsilon$. 
This means $W$ achieves its maximum value (i.e., $qe^\epsilon$) when the output point is exactly the same as the input one.

Let $\mathcal{C}=\{\vectorfont{z}|\|\vectorfont{z}\|_2\leq b\}$ and $r=\|\vectorfont{z}\|_2$.
According to $\iint_\mathcal{C} qe^{(1-\frac{r}{b})\epsilon}rdrd\theta=1-(4b+1)q$, we can obtain $q=\frac{\epsilon^2}{2\pi(e^\epsilon-1-\epsilon)b^2+4\epsilon^2 b+\epsilon^2}$.
As $\epsilon\to 0$, $q\to \frac{1}{\pi b^2+4b+1}$.
In this case, \solutionD{} degenerates into uniform random mechanism. It reports any value uniformly and randomly, without any utility guarantee.
As $\epsilon\to +\infty$, $q\to 0$, which means \solutionD{} reports the truthful value without any privacy protection.

We take \solutionD{} as one of our basic mechanisms and compare it with others in the experiment.

\section{The \solutionBTotalName{}}\label{BestMethod}
Although \solutionD{} achieves $\epsilon$-LDP, it requires a strong assumption that the probability within radius $b$ decreases with distance (similar to Geo-I~\cite{DBLP:conf/ccs/AndresBCP13}).
In this section, we remove this assumption and study the best probability distribution mechanism of \solutionGeneralName{} to maximize the accuracy of distribution estimation.
We call this best mechanism as \solutionBTotalName{} (\solutionB{}).
Besides, we give the method for choosing the high probability radius $b$ to further improve the estimation accuracy.

\revision{\subsection{Sliced Wasserstein Distance}\label{subsection:RT_SW}}
\revision{In order to get the best \solutionGeneralName{}, we need to use the closed form of $\mathcal{L}_2^2$-Wasserstein distance $W_2^2(D,\tilde{D})$ to deduce the relationship between the wave function $W$ and $W_2^2(D,\tilde{D})$.}
However, except for the 2-Dim normal distribution, there is no closed form for Wasserstein distance of 2-Dim data distribution when the data dimension $k>1$~\cite{panaretos2019statistical}, making the distribution analysis challenging.

\revision{To solve the above issue, we use sliced Wasserstein distance~\cite{DBLP:conf/nips/KolouriNSBR19} instead of Wasserstein distance.}
Sliced Wasserstein distance is a variant of Wasserstein distance. 
It achieves the measurement effect of Wasserstein distance in high-dimension while simplifying the calculation process.
The definition of sliced Wasserstein is based on Radon Transforms~\cite{helgason1980radon}.
Next we first introduce Radon transform in Definition~\ref{def_RT} and then give the definition of sliced Wasserstein distance in Definition~\ref{def:SWD}.

\begin{figure}[t!]\centering
	\subfigure[][{\small Radon transform}]{
		\scalebox{0.25}[0.25]{\includegraphics{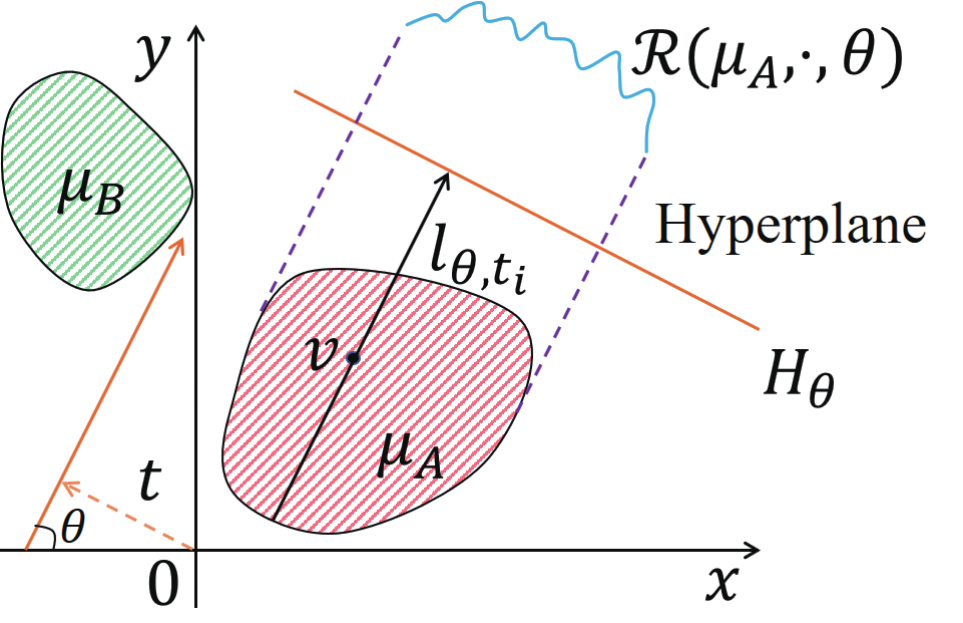}
			\label{subfig:Radon Transform}}}
	\subfigure[][{\small Wasserstein distance transform}]{
		\scalebox{0.25}[0.25]{\includegraphics{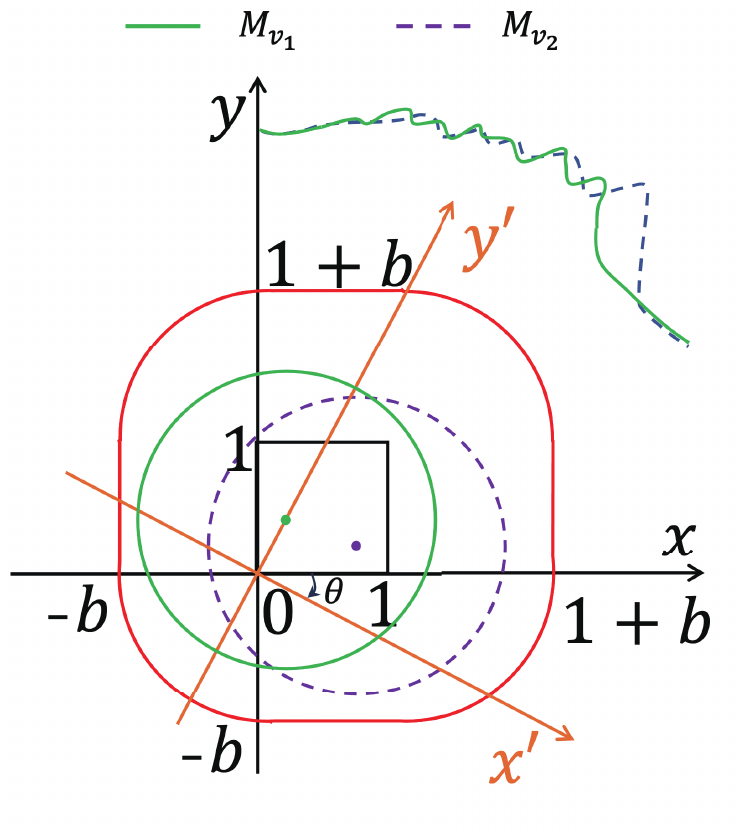}}
		\label{subfig:Construction_B}}
	\caption{\small Radon transform and sliced Wasserstein distance transform.}
	\label{fig:problem_norm}
\end{figure}

\begin{definition}(Radon Transform~\cite{helgason1980radon}).\label{def_RT}
	Let $\mathcal{L}^1(\mathbb{R}^k)=\{F:\mathbb{R}^k\to\mathbb{R}/\int_{\mathbb{R}^k}|F(x)|dx<\infty\}$ and $\mathbb{S}^{k-1}\subset\mathbb{R}^k$ the $k$-dimensional unit sphere.
	Let $\delta(\cdot)$ be the one-dimensional Dirac delta function.
	The standard Radon transform, defined as $\mathcal{R}$, maps a function $F\in \mathcal{L}^1(\mathbb{R}^k)$ to the infinite set of its integrals over the hyperplane of $\mathbb{R}^k$ and is defined as
	{\scriptsize $$\mathcal{R}(F,t,\theta)=\int_{\mathbb{R}^k}F(x)\delta(t- x\cdot\theta)dx,$$}
	where $(t,\theta)\in \mathbb{R}\times\mathbb{S}^{k-1}$.
\end{definition}

\begin{definition}(Sliced Wasserstein Distance~\cite{DBLP:conf/nips/KolouriNSBR19}).\label{def:SWD}
	Let $\mathcal{P}_p(\mathbb{R}^k)$ be the space of Borel probability measures on $\mathbb{R}^k$ with finite $p$-th moments.
	Given $\mu_A, \mu_B\in \mathcal{P}_p(\mathbb{R}^k)$, we define the $\mathcal{L}_k^p$-sliced Wasserstein distance as
	{\scriptsize $$SW_k^p(\mu_A,\mu_B)=\int_{\mathbb{S}^{k-1}}W_1^p(\mathcal{R}(\mu_A,\cdot,\theta),\mathcal{R}(\mu_B,\cdot,\theta))d\theta$$
	}
\end{definition}
The Radon transform maps a function $F(\mathbb{R}^k)$ to the infinite set of its integrals over hyperplanes in $\mathbb{R}^k$. The sliced Wasserstein distance in the domain $\mathbb{R}^k$ is defined as the integrals of the Wasserstein distance (between two distributions transformed by the Radon transform) over the unit sphere, denoted by $\mathbb{S}^{k-1}$, in $\mathbb{R}^k$. We give an example for Radon transforms and the sliced Wasserstein distance as follows.

Suppose the functions $\mu_A$ and $\mu_B$ are shown in Figure~\ref{subfig:Radon Transform} and the dimension $k=2$.
\revision{As for Radon transform, if we fix $\theta$ and set $t=t_i$, we can get the position and direction of integration (described as line $l_{\theta,t_i}$).}
When we want to get the Radon transform of $\mu_A$,
we can integrate all points in $\mu_A\cap l_{\theta,t_i}$ over $l_{\theta,t_i}$.
When we fix $\theta$ and alter $t$ from $-\infty$ to $+\infty$, we obtain the Radon transforms $\mathcal{R}(\mu_A,\cdot,\theta)$ for $\mu_A$, shown as the blue curve in Figure~\ref{subfig:Radon Transform}.
By altering $\theta$ from $0$ to $2\pi$ and $t$ from $-\infty$ to $+\infty$, we can compute the $\mathcal{L}_2^1$-sliced Wasserstein distance between $\mu_A$ and $\mu_B$ as $SW_2^1(\mu_A,\mu_B)=\int_{0}^{2\pi}W_1^1(\mathcal{R}(\mu_A,\cdot,\theta),\mathcal{R}(\mu_B,\cdot,\theta))d\theta$.

\subsection{Overview of \solutionB{}}

According to Reference~\cite{DBLP:conf/sigmod/Li0LLS20}, we can get the optimal mechanism by maximizing the $\mathcal{L}_2^2$ Wasserstein distance between $M_{\vectorfont{v}_1}$ and $M_{\vectorfont{v}_2}$ for any point pair $\vectorfont{v}_1$ and $\vectorfont{v}_2$. 
However, in the 2-Dim case, except for 2-Dim normal distributions, obtaining the closed-form solution of the $\mathcal{L}_2^2$ Wasserstein distance is difficult, which make the optimization objective hard.
To solve this problem, we substitute the $\mathcal{L}_2^1$-sliced Wasserstein distance for the $\mathcal{L}_2^2$-Wasserstein distance.
Thus, our optimization objective is transformed as \textit{Maximizing the $\mathcal{L}_2^1$ sliced Wasserstein between $M_{\vectorfont{v}_1}$ and $M_{\vectorfont{v}_2}$ for any different two points $\vectorfont{v}_1,\vectorfont{v}_2\in\mathcal{D}$.}
Figure~\ref{subfig:Construction_B} shows an example of this transform with a fix $\theta$.

\revision{Next, we describe \solutionBTotalName{} in Definition~\ref{solutionA_def}, and then prove that it is the best estimation among all kinds of \solutionGeneralName{}.
	\begin{definition}(\solutionBTotalName{}, \solutionB{}). \label{solutionA_def}
		A SAM is called a \solutionBTotalName{} if the $W$ function satisfies:
		\begin{equation}
			{\scriptsize
				\begin{aligned}
					W(\vectorfont{z})=
					\begin{cases}
						p,& \textrm{if}\;\;\|\vectorfont{z}\|_2\leq b   \\
						q,& otherwise
					\end{cases}
				\end{aligned}
			}
		\end{equation}
		where $p=\frac{e^\epsilon}{\pi b^2e^{\epsilon}+4b+1}$ and $q=\frac{1}{\pi b^2e^{\epsilon}+4b+1}$.
	\end{definition}
}

\revision{\solutionB{} is also a type of \solutionGeneralName{} with $W$'s value being constant in condition~\ref{cond:w_2}}.
In order to prove \solutionB{} is optimal, we need to get the partial derivative of the sliced Wasserstein in our optimization objective.
We give the partial derivative in Theorem~\ref{thrm:area} as follows.

\begin{theorem}\label{thrm:area}
	Given an angle $\theta\in[-\frac{\pi}{4},0]$ as the direction angle of projection line $l_{x'}$, $\vectorfont{v}_1, \vectorfont{v}_2\in \mathcal{D}$ as inputs to \solutionGeneralName{}, where $\Delta=(\vectorfont{v}_2-\vectorfont{v}_1)\cdot [\cos{\theta},\sin{\theta}]^T>0$, the partial derivative of sliced Wasserstein distance between the output distributions of \solutionGeneralName{} with respect to $\theta$ is $\Delta(1-(\pi b^2+4b+1)q)$.
\end{theorem}
\begin{proof}
	Let $\vectorfont{u}=(\cos{\theta},\sin{\theta})$.
	Given two different inputs $\vectorfont{v}_1, \vectorfont{v}_2\in \mathcal{D}$, where $(\vectorfont{v}_2-\vectorfont{v}_1)\cdot \vectorfont{u}^T=\Delta>0$.
	Let $M_{\vectorfont{v}_1}$ and $M_{\vectorfont{v}_2}$ are corresponding output distributions.
	We define a difference function as $DIFF(\vectorfont{z})$ in Equation~\eqref{AM1_Proof_DIFF}:
	\begin{equation}\label{AM1_Proof_DIFF}
		{\scriptsize
			\begin{aligned}
				DIFF(\vectorfont{z}) =
				\begin{cases}
					0,& \textrm{if}\;\;\vectorfont{z}\cdot \vectorfont{u}^T\leq -b\cos(\theta)   \\
					1-(\pi b^2+4b+1)q,& \textrm{if}\;\;\vectorfont{z}\cdot \vectorfont{u}^T\geq b\cos(\theta) \\
					\iint_{\mathcal{D}'}((W(\vectorfont{z})-q))d\vectorfont{z},& otherwise
				\end{cases}
			\end{aligned}
		}
	\end{equation}
	Let the output domain be $\mathcal{\tilde{D}}$, $\tilde{Y}'$ be the projection domain of $\mathcal{\tilde{D}}$ on axis $y'$.
	Then we can write the cumulative function on the $l_{x'}$ as 
	\begin{equation}\label{AM1_Proof_CMLTV}
		{\scriptsize
			\begin{aligned}
				P(M_{\vectorfont{v}}, \tilde{\vectorfont{v}}) &= h(\theta,\tilde{\vectorfont{v}}) + DIFF(\tilde{\vectorfont{v}}-\vectorfont{v})
			\end{aligned}
		}
	\end{equation}
	where {\scriptsize $h(\theta,\tilde{\vectorfont{v}})=\int_{\sin{\theta}-b\cos{\theta}}^{\tilde{\vectorfont{v}}\cdot\vectorfont{u}^T}\left(\int_{\tilde{Y}'}q dy'\right)dx'$
	}.
	Therefore, we have
	
	\begin{equation}\label{AM1_Proof_CCMLTV}
		{\scriptsize
			\begin{aligned}
				\int_{\sin{\theta}-b\cos{\theta}}^{(1+b)\cos{\theta}}P(M_{\vectorfont{v}}, \tilde{\vectorfont{v}})d\tilde{\vectorfont{v}} &= \int_{\sin{\theta}-b\cos{\theta}}^{(1+b)\cos{\theta}} h(\theta,\tilde{\vectorfont{v}})d\tilde{\vectorfont{v}} \\
				&+ \int_{\sin{\theta-b\cos{\theta}}}^{(1+b)\cos{\theta}} DIFF(\tilde{\vectorfont{v}}-\vectorfont{v})d\tilde{\vectorfont{v}}\\
				&= H(\theta) + \int_{-b\cos{\theta}}^{b\cos{\theta}} DIFF(\vectorfont{z})d\vectorfont{z}\\
				&+ (1-(\pi b^2+4b+1)q)(\cos{\theta}-\vectorfont{v}\cdot\vectorfont{u}^T)\\
			\end{aligned}
		}
	\end{equation}
	where $H(\theta)=\int_{\sin{\theta}-b\cos{\theta}}^{(1+b)\cos{\theta}} h(\theta,\tilde{\vectorfont{v}})d\tilde{\vectorfont{v}}$.
	According to the definition of sliced Wasserstein distance, we have
	\begin{equation}\label{AM1_Proof_SWD}
		{\scriptsize
			\begin{aligned}
				\frac{\partial SW_2^1(M_{\vectorfont{v}_1},M_{\vectorfont{v}_2})}{\partial\theta} &= \int_{\sin{\theta}-b\cos{\theta}}^{(1+b)\cos{\theta}}|P(M_{\vectorfont{v}_1}, \tilde{\vectorfont{v}})-P(M_{\vectorfont{v}_2}, \tilde{\vectorfont{v}})|d\tilde{\vectorfont{v}} \\
				&= (1-(\pi b^2+4b+1)q)\cdot\Delta
			\end{aligned}
		}
	\end{equation}
\end{proof}

According to Equation~\eqref{AM1_Proof_SWD}, to maximize $SW_2^1(M_{\vectorfont{v}_1},M_{\vectorfont{v}_2})$, we  need to minimize $q$.
Thus, we have Theorem~\ref{thrm:BestMechanism} as follows.

\begin{theorem}\label{thrm:BestMechanism}
	For any fixed value $b$ and $\epsilon$, the minimum $q$ for 2-norm mechanism is $q=\frac{1}{\pi b^2e^\epsilon+4b+1}$. This minimum can be achieved if and only if the mechanism is \solutionB{}.
\end{theorem}

\begin{proof}
	For any point $\vectorfont{v}\in \mathcal{D}$, let $\mathcal{\tilde{D}}_{\vectorfont{v}}=\{\vectorfont{\tilde{v}}|\parallel \vectorfont{\tilde{v}} - \vectorfont{v}\parallel_1\leq b\}$. Then the area of $\mathcal{\tilde{D}}_{\vectorfont{v}}$ is $\pi b^2$. Therefore, we have $\iint_{\mathcal{\tilde{D}}_{\vectorfont{v}}}W(\vectorfont{\tilde{v}} - \vectorfont{v})\leq \pi b^2\cdot e^\epsilon q$.
	And we have
	\begin{equation}\label{AM1_Proof_LOWERBOUND}
		{\scriptsize
			\begin{aligned}
				1   &= (4b+1)q+\iint_{\mathcal{\tilde{D}}_{\vectorfont{v}}}W(\vectorfont{\tilde{v}} - \vectorfont{v}) \\
				&\leq (4b+1)q + \pi b^2e^\epsilon q \\
				&= (\pi b^2e^\epsilon+4b+1)q. \\
			\end{aligned}
		}
	\end{equation}
	Thus, we have $q \geq \frac{1}{\pi b^2e^\epsilon+4b+1}$.
\end{proof}

\subsection{Choosing Radius b}\label{subsection:b_analysis}

In our \solutionB{}, a value within a distance of $b$ from the true value is reported with a probability that is $e^\epsilon$ times as large as the one outside of $b$. The optimal choice of $b$ depends on the privacy parameter $\epsilon$.

Intuitively, as $\epsilon$ approaches infinity, $b$ needs to approach $0$ to fully recover the input distribution. Additionally, when the probability density of the private distribution is concentrated at one point, a smaller $b$ is suitable, whereas when the probability has a more evenly distributed density, a larger $b$ is appropriate.
\revision{However, we do not know the distribution of the private points.
	There is a silver lining that we can choose a $b$ value independent of the distribution while still performing reasonably well over different distributions.
	Similar to the case of one dimension~\cite{DBLP:conf/sigmod/Li0LLS20},} we choose $b$ by maximizing the upper bound of mutual information between the input and output of our \solutionB{} in the 2-Dim case.

\noindent
\revision{
	\textbf{Unit Side Length Input.}}
We consider the case where the input domain is a unit square. Let $\vectorfont{V}$ and $\tilde{\vectorfont{V}}$ be the random variables representing the input and output in our \solutionB{}. We can express the mutual information as the difference between the differential entropy of $\tilde{\vectorfont{V}}$ and the conditional differential entropy of $\vectorfont{V}$ and $\tilde{\vectorfont{V}}$:
{\scriptsize $$I(\vectorfont{V},\tilde{\vectorfont{V}})=h(\tilde{\vectorfont{V}})-h(\tilde{\vectorfont{V}}|\vectorfont{V}).$$
}
$h(\tilde{\vectorfont{V}})$ is maximized when $\tilde{\vectorfont{V}}$ is uniformly distributed on $\tilde{\vectorfont{D}}$.
For \solutionB{}, we have
\begin{equation}\label{Chosen_b2}
	{\scriptsize
		\begin{aligned}
			I(\vectorfont{V},\tilde{\vectorfont{V}}) &\leq h(\tilde{\vectorfont{U}})-h(\tilde{\vectorfont{V}}|\vectorfont{V})\\
			&= \log{(\pi b^2+4b+1)} + (\pi b^2p\log{p}+(4b+1)q\log{q})\\
			&= \log{(\frac{\pi b^2+4b+1}{\pi b^2e^\epsilon+4b+1})}+\epsilon\log{e}-\frac{(4b+1)\epsilon\log{e}}{\pi b^2e^\epsilon+4b+1}.\\
		\end{aligned}
	}
\end{equation}
We denote the expression on the right side in Equation~\eqref{Chosen_b2} as $g(b)$.
Then we have:
\begin{equation}\label{Partial_Chosen_b2}
	{\scriptsize
		\begin{aligned}
			\frac{dg(b)}{db} &= \frac{2\pi b(2b+1)(-\pi e^\epsilon m_1 b^2+4m_2 b + m_2)}{(\pi b^2+4b+1)(\pi b^2e^\epsilon+4b+1)^2\ln{2}}\\
		\end{aligned}
	}
\end{equation}
where $m_1=e^\epsilon-1-\epsilon$ and $m_2=1-e^\epsilon+\epsilon e^\epsilon$.
Because $\epsilon$ and $b$ are both positive, when $b=\frac{2m_2+\sqrt{4m_2^2+\pi e^\epsilon m_1 m_2}}{\pi e^\epsilon m_1}$, it achieves maximum.
We can see that when $\epsilon\to 0$, $b\to \frac{2+\sqrt{4+\pi}}{\pi}$, and when $\epsilon\to +\infty$, $b\to 0$.

\noindent
\revision{
	\textbf{General Side Length Input.}}
When the input domain is a square with length $L$, the mutual information of $\tilde{\vectorfont{V}}$ and $\vectorfont{V}$ can be express as follow:
\begin{equation}\label{Chosen_b_d}
	{\scriptsize
		\begin{aligned}
			I(\vectorfont{V},\tilde{\vectorfont{V}}) &\leq \log{(\pi b^2+4Lb+L^2)} + (\pi b^2p\log{p}+(4Lb+L^2)q\log{q})\\
			&= \log{(\frac{\pi b^2+4Lb+L^2}{\pi b^2e^\epsilon+4Lb+L^2})}+\frac{\pi b^2e^\epsilon\epsilon\log{e}}{\pi b^2e^\epsilon+4Lb+L^2}.\\
		\end{aligned}
	}
\end{equation}
And we have
\begin{equation}\label{Partial_Chosen_b2_general}
	{\scriptsize
		\begin{aligned}
			\frac{dg(b)}{db} &= \frac{2\pi Lb(2b+L)(-\pi e^\epsilon m_1 b^2+4dm_2 b + L^2m_2)}{(\pi b^2+4Lb+L^2)(\pi b^2e^\epsilon+4Lb+L^2)^2\ln{2}}.\\
		\end{aligned}
	}
\end{equation}
Let $m_1=e^\epsilon-1-\epsilon$ and $m_2=1-e^\epsilon+\epsilon e^\epsilon$.
Let $\frac{dg(b)}{db}=0$, we can get the best $b=\frac{2m_2+\sqrt{4m_2^2+\pi e^\epsilon m_1 m_2}}{\pi e^\epsilon m_1}\cdot L$.

\section{Bucketizing and Post-processing}\label{Implement}
\revision{When we use any \solutionGeneralName{} (e.g., \solutionD{} or \solutionB{}) in real-world scenarios, it is impossible to count the frequency of all types of points because there is infinite number of points in any continuity ranges (e.g., $\mathcal{D}$ and $\mathcal{\tilde{D}}$).
	Thus we need to bucketize the input/output domain into grids and execute \solutionGeneralName{} under grid domain.
}
Additionally, we provide the method of post-processing under the grid condition.
Finally, we give the total algorithms of solving \problemDefineSimpleName{}.

\subsection{Bucketizing}
To facilitate the reconstruction of the distribution, we need to divide the plane into grids and use our \solutionB{} on this grid plane. In other words, the problem is converted into estimating the histogram distribution on a $2$-Dim plane using \solutionB{}.

Let $g$ be the length of a grid cell.
Let $G$ be the grid input domain and $\tilde{G}$ be the grid output domain.
We denote the side length of the grid as $d=\lfloor\frac{L}{g}\rfloor$ and the high probability radius in the grid as $\hat{b}=\lfloor\frac{b}{g}\rfloor$.
Then, the coordinate unit is reset to the side length of a grid cell, and we use the central point of a cell to represent its position.
For example, in Figure~\ref{subfig:Reconstruct_A_1}, the index of the cell $a_0$ is $(0,0)$, and the index of the cell to the right of $a_0$ is $(1,0)$.

\begin{figure}[t!]\centering\vspace{2ex}
	\subfigure[][{\small Non-shrunken areas}]{
		\scalebox{0.35}[0.35]{\includegraphics{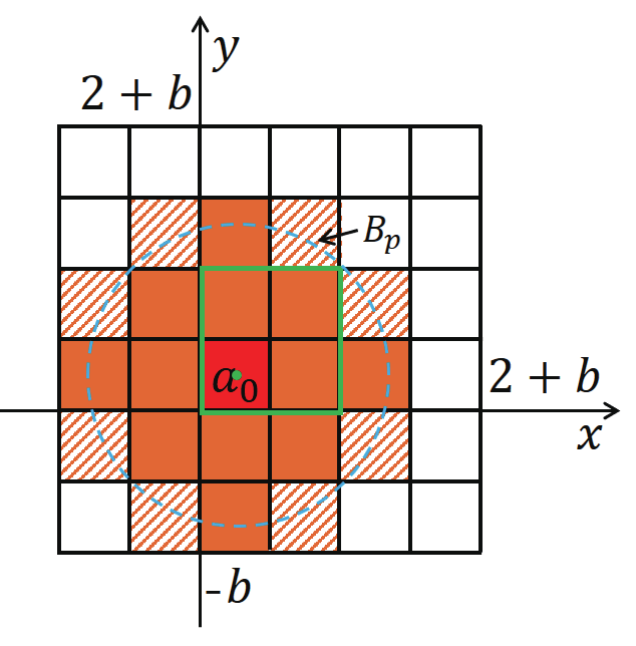}}
		\label{subfig:Reconstruct_A_1}}\hfill
	\subfigure[][{\small Shrunken areas}]{
		\scalebox{0.35}[0.35]{\includegraphics{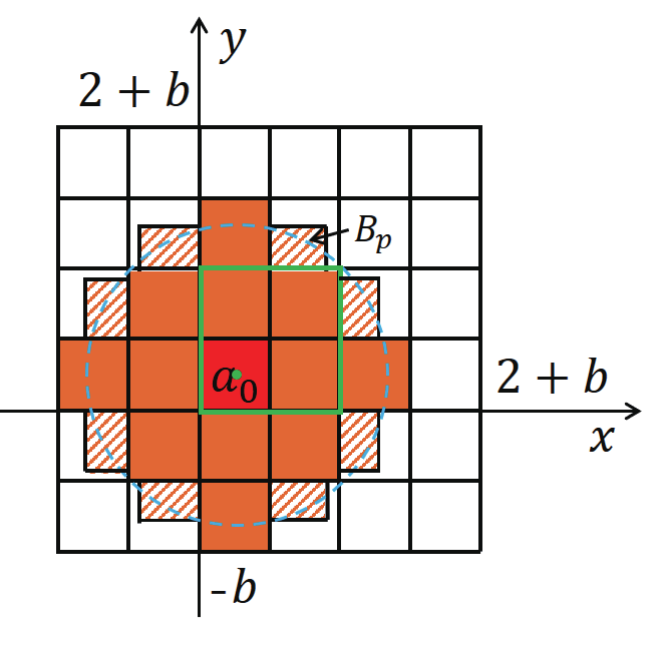}}
		\label{subfig:Reconstruct_A_2}}
	\caption{\small Non-shrunken/Shrunken areas in grid division.}
	\label{fig:Reconstruct_A}
\end{figure}

Our \solutionB{} for grids is defined in Equation~\eqref{solutionB_G_def}.

\begin{equation}\label{solutionB_G_def}
	{\scriptsize
		\begin{aligned}
			\forall v\in G, \tilde{v}\in\tilde{G}, M_{\vectorfont{v}}(\tilde{\vectorfont{v}})=
			\begin{cases}
				\hat{p},& \textrm{if}\;\;\|\tilde{\vectorfont{v}}-\vectorfont{v}\|_2\leq \hat{b},   \\
				\hat{q},& otherwise.
			\end{cases}
		\end{aligned}
	}
\end{equation}

Next, we decide how to calculate $\hat{p}$ and $\hat{q}$ to approximate the true values.

As is shown in Figure~\ref{fig:Reconstruct_A}, given an input cell $a_0$ (the red rectangle), the blue dotted line (denoted as $B_p$) represents the border of high probability reporting. Based on the positional ordinal relationship between $B_p$ and the output cells, the output domain can be divided into three areas:
\begin{enumerate}[label=(\arabic*)]
	\item \emph{the pure high probability area $A_p$} where the center of each cell is in or on $B_p$;
	\item \emph{the pure low probability area $A_q$} where each cell neither intersects with $B_p$ nor locates in $B_p$;
	\item \emph{the mixed probability area $A_m$} where each cell intersects with $B_p$, however, the center point is out of $B_p$.
\end{enumerate}
All these areas are shown as the orange cells, the white cells and shaded cells, respectively. Each cell $a^{(i)}$ in $A_m$ can be further divided into high probability part $a^{(i)}_p$ and low probability part $a^{(i)}_p$, respectively. We combine all high probability areas ($A_{m,p}=\sum_{i}a^{(i)}_p$) in $A_m$ with $A_p$ to form the total high probability area $A_H$. Similarly, we combine all low probability areas ($A_{m,q}=\sum_{i}a^{(i)}_q$) in $A_m$ with $A_q$ to form the total low probability area $A_L$.

We consider the area size of each cell to be 1 (i.e., $S_a=1$).
To determine $\hat{p}$ and $\hat{q}$, we need to solve two problems:
\begin{enumerate}[label=(\arabic*)]
	\item How can we determine the area size of $a^{(i)}_p$ for each cell in $A_m$ to satisfy $\epsilon$-LDP ?\label{prblm_1}
	\item What is the area size of $A_H$ and $A_L$ ?\label{prblm_2}
\end{enumerate}
To solve Problem~\ref{prblm_1}, we first determine the center of $a^{(i)}_p$ (denoted as $C_N$) by intersecting $B_p$ and the line between $a^{(i)}$'s center and $B_p$'s center. Then, we construct $a^{(i)}_p$ as a rectangle centered at $C_N$ satisfying $a^{(i)}_p$'s left and bottom borders overlap $a^{(i)}_p$'s left and bottom borders respectively.
To solve Problem~\ref{prblm_2}, we first partition $A_H$ and $A_L$ into several parts and reconstruct them as different types of cells. Then, we count these cell areas by category. We introduce these methods in detail below.

\begin{figure}[t!]\centering\vspace{2ex}
	\scalebox{0.32}[0.32]{\includegraphics{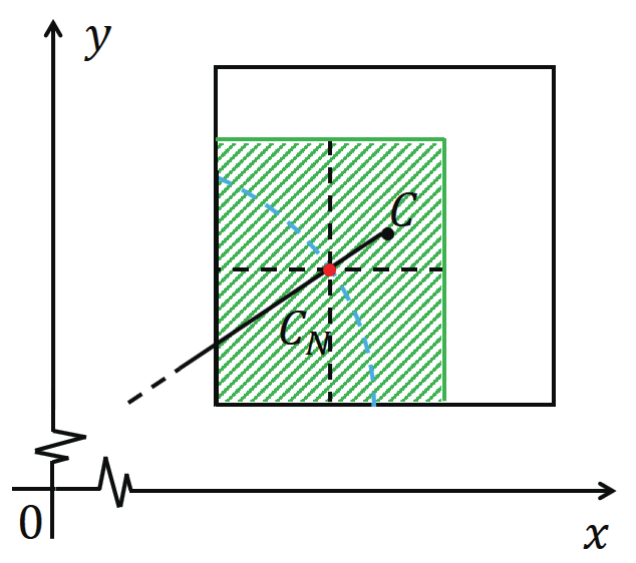}}
	\caption{\small The process of border shrinkage in discrete \solutionB{}.}
	\label{fig:Approximate_A_1}
\end{figure}

As for the Problem~\ref{prblm_1}, we determine whether the center $o^{(i)}$ of each cell $a^{(i)}\in A_m$ is inside $B_p$.
If $o^{(i)}$ is outside of $B_p$, we create a new shrunken rectangle for $a^{(i)}$ and set it as $a^{(i)}_p$ while the remain part $a^{(i)}\setminus a^{(i)}_p$ is set as $a^{(i)}_q$. The method for constructing the new shrunken rectangle is as follows.

Suppose $B_p$ crosses a grid cell $a$ (with its central point noted as $C$) as shown in Figure~\ref{fig:Approximate_A_1}, with the blue dot arc $ARC$ indicating the intersection. We connect the center of $B_p$ and the point $C$ to obtain the intersection point $C_N$ on $ARC$. We define $C_N$ as the center of the shrunken rectangle $a^{(i)}_p$, and construct $a^{(i)}_p$ as shown in the green shaded part.

Next we give Theorem~\ref{thrm:shrankArea} to calculate the area size of the shrunken rectangle $a^{(i)}_p$.

\begin{theorem}\label{thrm:shrankArea}
	Given a circle $B_{p}$ with central cell index $(0,0)$, a radius $\hat{b}$ and any cell $a$, the area size of $a$'s shrunken cell $a_p$ is
	$S_{a_p}=4(\delta\cdot x + \frac{1}{2})(\delta\cdot y + \frac{1}{2})$, where $\delta=\frac{\hat{b}}{\sqrt{x^2+y^2}}-1$, and $a$ is any cell that intersects with $B_{p}$, whose central point $(x,y)$ is outside the range of $B_{p}$.
\end{theorem}
\begin{proof}
	As shown in Figure~\ref{fig:Approximate_A_1}, suppose the index of cell $a$ is $(x,y)$, then the index of $a_p$ is $(\hat{b}\cdot\frac{x}{\sqrt{x^2+y^2}}, \hat{b}\cdot\frac{y}{\sqrt{x^2+y^2}})$.
	The line of $a$'s left boundary is $X=x-\frac{1}{2}$ and the line of $a$'s bottom boundary is $Y=y-\frac{1}{2}$.
	Thus, $S_{a_p}=4(\hat{b}\cdot\frac{x}{\sqrt{x^2+y^2}}-(x-\frac{1}{2}))(\hat{b}\cdot\frac{y}{\sqrt{x^2+y^2}}-(y-\frac{1}{2}))$.
	Let $\delta=\frac{\hat{b}}{\sqrt{x^2+y^2}}-1$.
	Then we have $S_{a_p}=4(\delta\cdot x + \frac{1}{2})(\delta\cdot y + \frac{1}{2})$.
\end{proof}

As for the problem~\ref{prblm_2}, we decompose $A_H=A_p+A_{m,p}$ and $A_L=A_q+A_{m,q}$.
We need to calculate the area size of $A_q$, $A_p$ and $A_m$ (i.e., $A_{m,p}$ + $A_{m,q}$).
Next, we give Theorem~\ref{thrm:A_q},~\ref{thrm:SOCount} and~\ref{thrm:SICount} to calculate the sizes of these three areas.

\begin{theorem}\label{thrm:A_q}
	For any square input domain $\mathcal{D}$ with integer side length $d$ and any integer high probability radius $\hat{b}$, the area size of pure low probability area $A_q$ is $d^2+4\hat{b}d-4\hat{b}-1$.
\end{theorem}
\begin{proof}
	Please refer to details of Theorem~\ref{thrm:A_q} in Appendix~B1.
\end{proof}

Theorem~\ref{thrm:A_q} gives the method to calculate the area size of $A_q$.
As for $A_p$ and $A_m$, according to the centripetal symmetry and axial symmetry of a circle, we only need to analyze the part within angle $[0,\frac{\pi}{4}]$.
Figure~\ref{subfig:Approximate_A_2} shows the conditions that $\hat{b}$ is $1,2,...,7$.
The cells in directions of $0$ and $\frac{\pi}{4}$ are in yellow while others are in green.

We define $E_{\hat{b},\theta}$ as the cells in direction $\theta$ within the radial range of $\hat{b}$.
We define $E_{\hat{b},(\theta_1,\theta_2)}$ as the cells within the radial range of $\hat{b}$ and the direction range of $(\theta_1,\theta_2)$.
\revision{We define \textit{strict quarter $A_m$} (denoted as $E_{\hat{b},(0,\frac{\pi}{4})}^{(m)}$) as the cells belonging to $A_m$ in $E_{\hat{b},(0,\frac{\pi}{4})}$,
	and \textit{strict quarter $A_p$} (denoted as $E_{\hat{b},(0,\frac{\pi}{4})}^{(p)}$) as the cells belonging to $A_p$ in $E_{\hat{b},(0,\frac{\pi}{4})}$.
	Similar to $A_m$, the strict quarter $A_m$ can be divided into the high probability part and low probability part. 
	We call these two parts as \textit{strict quarter $A_{m,p}$} and and \textit{strict quarter $A_{m,q}$} respectively.}
For example, in Figure~\ref{subfig:Approximate_A_2},
$E_{\hat{7},\frac{\pi}{4}}=\{a_0\}\cup\{D_i|i\in[5]\}$,
$E_{7,(0,\frac{\pi}{4})}^{(m)}=\{A_1,A_2,A_3,A_4\}$, the quantity of $E_{7,(0,\frac{\pi}{4})}^{(p)}$ is  $|E_{7,(0,\frac{\pi}{4})}^{(p)}|=|\{B_i|i\in[13]\}|=13$.

\revision{In order to get the area size of strict quarter $A_m$ (i.e., strict quarter $A_{m,p}$ and strict quarter $A_{m,q}$), we need to know each cell's index in $E_{\hat{b},(0,\frac{\pi}{4})}^{(m)}$. 
	Based on these cell indexes and Theorem~\ref{thrm:shrankArea}, we can calculate each cell's shrunken area size and remaining area size in $E_{\hat{b},(0,\frac{\pi}{4})}^{(m)}$ (i.e., strict quarter $A_{m,p}$ and strict quarter $A_{m,q}$).
	We present Theorems~\ref{thrm:SOCount} to get the cell indexes in $E_{\hat{b},(0,\frac{\pi}{4})}^{(m)}$ as follows.}

\begin{theorem}\label{thrm:SOCount}
	Given a positive integer $\hat{b}$, the quantity of 
	{\scriptsize $E_{\hat{b},(0,\frac{\pi}{4})}^{(m)}$}
	is 
	$\lceil\frac{\hat{b}}{\sqrt{2}}-\frac{1}{2}\rceil-\lfloor\frac{r}{\hat{b}}\rfloor$, 
	where $r=\sqrt{r_1^2+1+\sqrt{2}r_1}$ and  $r_1=\lfloor\frac{\hat{b}}{\sqrt{2}}-\frac{1}{2}\rfloor\cdot\sqrt{2}+\frac{1}{\sqrt{2}}$.
	The index of each cell in $E_{\hat{b},(0,\frac{\pi}{4})}^m$ is  $(\lceil\sqrt{\hat{b}^2-(i-\frac{1}{2})^2}-\frac{1}{2}\rceil,i)$ for  $i\in[|E_{\hat{b},(0,\frac{\pi}{4})}^{(m)}|]$.
\end{theorem}
\begin{proof}
	Please refer to details of Theorem~\ref{thrm:SOCount} in Appendix~B2.
\end{proof}

\revision{In order to get the area size of strict quarter $A_p$, we need to count the cell number in this area.
	Theorem~\ref{thrm:SICount} gives the method to calculate this count.
}

\begin{figure}[t!]\centering
	\includegraphics[width=0.23\textwidth]{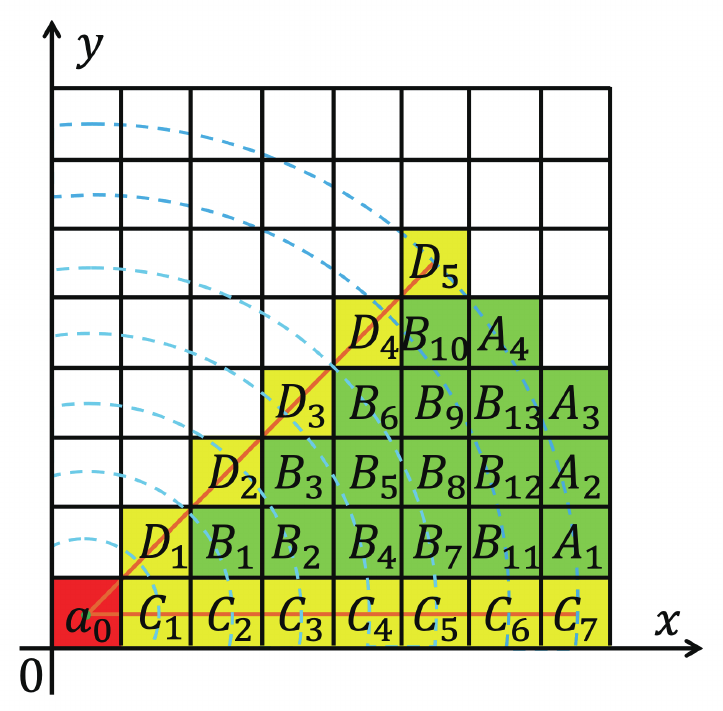}
	\caption{\small Quarter of the total area.}
	\label{subfig:Approximate_A_2}
\end{figure}

\begin{theorem}\label{thrm:SICount}
	Given a positive integer $\hat{b}$, the quantity of $E_{\hat{b},(0,\frac{\pi}{4})}^{(p)}$ is
	{\scriptsize $\frac{1}{2}\lceil\frac{\hat{b}}{\sqrt{2}}-\frac{1}{2}\rceil(\lceil\frac{\hat{b}}{\sqrt{2}}-\frac{1}{2}\rceil-2|E_{\hat{b},(0,\frac{\pi}{4})}^{(m)}|-1) +\sum_{i=1}^{|E_{\hat{b},(0,\frac{\pi}{4})}^{(m)}|}\lceil\sqrt{\hat{b}^2-(i-\frac{1}{2})^2}-\frac{1}{2}\rceil$}.
\end{theorem}
\begin{proof}
	Please refer to details of Theorem~\ref{thrm:SICount} in Appendix~B3.
\end{proof}

According to Theorems~\ref{thrm:shrankArea}, \ref{thrm:SOCount}, and \ref{thrm:SICount}, we can calculate $\hat{p}_2$ and $\hat{q}_2$ as follows.

Let $S_a^{(m,p)}$ be the shrunken area size of $a\in E_{\hat{b},(0,\frac{\pi}{4})}^{(m)}$, which can be calculated by Theorem~\ref{thrm:shrankArea}.
Similarly, let $S_{\frac{\pi}{4}}^{(m,p)}$ denote the shrunken area size of cell $a\in E_{\hat{b}}^{\frac{\pi}{4}}\cap A_m$. According to Theorem~\ref{thrm:shrankArea}, we have
\begin{equation}\label{A_pi_4_define}
	{\scriptsize
		\begin{aligned}
			S_{\frac{\pi}{4}}^{(m,p)}=
			\begin{cases}
				4(b'_{\frac{\pi}{4}}-\hat{b}_{\frac{\pi}{4}})^2,& \textrm{if}\;\;b'_{\frac{\pi}{4}}-\hat{b}_{\frac{\pi}{4}}< \frac{1}{2},   \\
				1,& otherwise
			\end{cases}
		\end{aligned}
	}
\end{equation}
where {\scriptsize $b'_{\frac{\pi}{4}}=\frac{\hat{b}}{\sqrt{2}}-\frac{1}{2}$} and 
{\scriptsize $\hat{b}_{\frac{\pi}{4}}=\lfloor b'_{\frac{\pi}{4}}\rfloor$}.

Finally, we can calculate the probabilities $\hat{p}$ and $\hat{q}$ as:
{\scriptsize$$\hat{p}=\frac{e^\epsilon}{S_H\cdot e^\epsilon+S_L},\;\;\;\; \hat{q}_2=\frac{1}{S_H\cdot e^\epsilon+S_L},$$}
where 
{\scriptsize $S_H=1+4(\hat{b}+\hat{b}_{\frac{\pi}{4}}+S_{\frac{\pi}{4}}^{(m,p)})+8(|E_{\hat{b},(0,\frac{\pi}{4})}^{(p)}|+\sum\limits_{a\in E_{\hat{b},(0,\frac{\pi}{4})}^{(m)}}S_a^{(m,p)})$} 
and
{\scriptsize $A_L=A_q+4(1-S_{\frac{\pi}{4}}^{(m,p)})+8\sum\limits_{a\in E_{\hat{b},(0,\frac{\pi}{4})}^{(m)}}(1-S_a^{(m,p)})$}.

Regarding the discretization of \solutionD{}, the high probability areas can be divided into $\hat{b}$ fan rings $\{FR_j\}_{j=1}^{\hat{b}}$.
For any unit cell in $FR_j$, the reported probability is $p_j^{(I)}=qe^{1-\frac{j-1}{\hat{b}}}$.
For a unit cell $a$ on the border of $FR_{j-1}$ and $FR_{j}$, the reported probability is $p_j^{(B)}=S_{a}^{(p)}\cdot p_{j-1}^{(I)}+(1-S_{a}^{(p)})\cdot p_{j}^{(I)}$, where $S_{a}^{(p)}$ is the shrunken area size of $a$.
For more details on discretization of \solutionD{}, please refer to Appendix~A.

\begin{algorithm}[t!]
	\DontPrintSemicolon
	\caption{\solutionB{} Processing Framework}
	\label{alg_PSDE}
	\KwIn{original data point set $X$, square range $L\times L$, cell side length $g$, privacy budget $\epsilon$}
	\KwOut{distribution map $R$}
	{Split square range into $\lceil\frac{L}{g}\times\frac{L}{g}\rceil$ grids with index set $\mathcal{I}=[0:\lceil\frac{L}{g}\rceil-1]\times[0:\lceil\frac{L}{g}\rceil-1]$;}\\
	{Calculate the \revision{noisy} domain index set $\hat{\mathcal{I}}$;}\\
	{Initialize \revision{noisy} map $NR$ by setting items as $(\hat{i},0)$ for each $\hat{i}\in\hat{\mathcal{I}}$;}\label{noise_domain}\\
	\For{each point $x$ in $X$}{
		{Get the grid index $\mathcal{I}(x)$;}\label{project_to_cell}\\
		{$\hat{\mathcal{I}}_x\gets\emph{GridAreaResponse}(\mathcal{I}(x))$;}\label{project_to_noise_cell}\\
		{$NM(\mathcal{I}(x))\gets NM(\mathcal{I}(x))+1$;}\\
	}
	{$R\gets\emph{PostProcess}(NR, I)$;}\\
	\Return $R$;
\end{algorithm}

\begin{algorithm}[t!]
	\DontPrintSemicolon
	\caption{GridAreaResponse}
	\label{alg_solutionB}
	\KwIn{original grid index $i$}
	\KwOut{\revision{noisy} grid index $\hat{i}$}
	{Find the pure high probability cell index set $\hat{\mathcal{I}}_{p}(i)$ and calculate its total area size $S_{p}$;}\\
	{Find the pure low probability cell index set $\hat{\mathcal{I}}_{q}(i)$ and calculate its area size $S_{q}$;}\\
	{Find the high probability border cell set $\hat{\mathcal{I}}_{m}(i)$ and calculate the sum shrunken area $S_{m,p}=\sum_{\hat{i}\in\hat{\mathcal{I}}_{m}}S_{m,p}(\hat{i})$ and the complement area $\overline{S}_{m,p}=\sum_{\hat{i}\in\hat{\mathcal{I}}_{m}}\overline{S}_{m,p}(\hat{i})$;}\\
	{Set value list $vl=<A_{PL},\overline{S}_{m,p}, S_{m,p}, A_{PH}>$;}\\
	{Set weighted list $wl=<1, 1, e^\epsilon, e^\epsilon>$;}\\
	{Sample $ind$ as $i$ with probability as $p_i=\frac{vl_i\cdot wl_i}{\sum_{j=1}^{4} vl_j\cdot wl_j}$;}\label{weighted_sample_1}\\
	\If{$ind=1$}{
		{$\hat{i}\xleftarrow{\$} \hat{\mathcal{I}}_{q}(i)$;}\label{random_1}\\
	}\ElseIf{$ind=4$}{
		{$\hat{i}\xleftarrow{\$} \hat{\mathcal{I}}_{p}(i)$;}\label{random_2}\\
	}\Else{
		{Set $vl=<ws_1,...,ws_n>$ for each $ws_j=\overline{sa}_j+sa_j\cdot e^\epsilon$;}\label{reset_area}\\
		{Set $wl=<1,...,1>$ with $n$ elements;}\\
		{Sample $ind$ as $i$ with probability as $p_i=\frac{vl_i\cdot wl_i}{\sum_{j=1}^{4} vl_j\cdot wl_j}$;}\label{weighted_sample_2}\\
		{$\hat{i}\gets$ cell with $ind$ in $vl$;}\\
	}
	\Return $\hat{i}$;
\end{algorithm}

\subsection{The \problemDefineSimpleName{} Processing Algorithm}
We give the processing framework for our \solutionB{} shown in Algorithm~\ref{alg_PSDE}. 
The input square range with an area size of $L\times L$ will be divided into grids. For each point in this area, we first project it onto a grid cell (Line~\ref{project_to_cell}), and then randomize the cell into a random \revision{noisy} cell (Line~\ref{project_to_noise_cell}) using \emph{GridAreaResponse}. All the points in each \revision{noisy} cell will be counted and stored in the \revision{noisy} map $NM$. Finally, we obtain the distribution estimation using \emph{PostProcess}.

In Algorithm~\ref{alg_PSDE}, the \emph{GridAreaResponse} process is to pick a randomized cell index satisfying $\epsilon$-LDP.
The \emph{PostProcess} process is to handle the values to obtain an accurate estimation distribution, which is the Expectation-Maximization (EM)~\cite{DBLP:conf/sigmod/Li0LLS20} Algorithm.
Next, we give the processes of \emph{GridAreaResponse} in Algorithm~\ref{alg_solutionB}.

In \emph{GridAreaResponse}, a cell point in range $b$ will have a higher probability of being responded to while those outside of it will have a lower probability.
Specifically, the areas are divided into three parts: the pure low probability area, the mixed probability area, and the pure high probability area.
Given an original grid index $i$, the algorithm calculates the area size of high probability part $\hat{\mathcal{I}}_{p}(i)$, the low probability part $\hat{\mathcal{I}}_{q}(i)$ and the mixed probability area $\hat{\mathcal{I}}_{m}(i)$.
The cells that are crossed by the circle centered at cell $(0,0)$ with radius $b$ make up $\hat{\mathcal{I}}_{m}(i)$.
All the cells in the mixed probability area need to be further divided into two parts: the shrunken part and the remain part. Therefore, there are four parts of the area that can be chosen as a candidate sample domain. The algorithm uses a weighted sample (Line~\ref{weighted_sample_1}) to determine which part to choose. The value $ind=1$ refers to the choice of the low probability part, and $ind=4$ refers to the choice of the high probability part. Both of these two cases use the uniform sample to choose $\hat{i}$ (Line~\ref{random_1} and Line~\ref{random_2}). When it comes to the case of border area containing $n$ cells, rather than using the uniform sample, the sampling probability needs to be proportional to the weighted area size $ws_j$ for each cell in the mixed probability area (Line~\ref{reset_area}). After that, the weighted sampling algorithm is used with identical weight for each cell to sample the result response cell (Line~\ref{weighted_sample_2}).

\emph{Time and Memory Complexity Analysis}.
Let $n$ be the number of users. Let $g$ be the grid number of the input domain.
The time complexity of GridAreaResponse algorithms is $O(g)$.
Let $m$ be the repeat times before converging in PostProcess algorithm.
The time complexity of PostProcess is $O(nk)$.
Therefore, the time complexity of \solutionB{} Processing Framework is $O(ng+nk)$.
The memory complexity of GridAreaResponse and PostProcess algorithms are $O(g)$.
Therefore, the memory complexity of \solutionB{} Processing Framework is $O(g)$.

\section{Experimental Evaluation}\label{Experiment}

In this section, we compare the accuracy of our mechanisms with state-of-the-art methods across various parameters. Our experiments aim to determine which method achieves the smallest Wasserstein distances between the real and recovered obfuscated density distributions under equivalent privacy levels or grid sizes.

\subsection{Experimental Setup}
\begin{table}[t!]
	\begin{center}
		{\small\scriptsize \vspace{-2ex}
			\caption{\small The range and data points of Data sets.} \label{dataset:parts}
			\begin{tabular}{|c|cc|cc|}\hline
				\multirow{2}{*}{} & \multicolumn{2}{c|}{Chicago Crimes}                                                                                            & \multicolumn{2}{c|}{NYC Green Taxies}                                                                                            \\ \cline{2-5}
				& \multicolumn{1}{c|}{Range}                                                                                        & Point size & \multicolumn{1}{c|}{Range}                                                                                          & Point size \\ \hline
				Part A            & \multicolumn{1}{c|}{\begin{tabular}[c]{@{}c@{}}{[}41.72$^{\circ}$,41,81$^{\circ}${]}\\      $\times${[}-87.68$^{\circ}$,-87.59$^{\circ}${]}\end{tabular}} & 216,595    & \multicolumn{1}{c|}{\begin{tabular}[c]{@{}c@{}}{[}40.65$^{\circ}$, 40.75$^{\circ}${]}\\      $\times${[}-73.84$^{\circ}$, -73.74$^{\circ}${]}\end{tabular}} & 10,561     \\ \hline
				Part B            & \multicolumn{1}{c|}{\begin{tabular}[c]{@{}c@{}}{[}41.82$^{\circ}$,41.91$^{\circ}${]}\\      $\times${[}-87.73$^{\circ}$,-87.64$^{\circ}${]}\end{tabular}} & 173,552    & \multicolumn{1}{c|}{\begin{tabular}[c]{@{}c@{}}{[}40.65$^{\circ}$,40.74$^{\circ}${]}\\      $\times${[}-73.95$^{\circ}$,-73.86$^{\circ}${]}\end{tabular}}   & 42,195     \\ \hline
				Part C            & \multicolumn{1}{c|}{\begin{tabular}[c]{@{}c@{}}{[}41.92$^{\circ}$,41.99$^{\circ}${]}\\      $\times${[}-87.77$^{\circ}$,-87.70$^{\circ}${]}\end{tabular}} & 69,068     & \multicolumn{1}{c|}{\begin{tabular}[c]{@{}c@{}}{[}40.82$^{\circ}$,40.89$^{\circ}${]}\\      $\times${[}-73.90$^{\circ}$,-73.83$^{\circ}${]}\end{tabular}}   & 9,186      \\ \hline
			\end{tabular}
		}
	\end{center}
\end{table}

\noindent
\textbf{Data sets.} We demonstrate all above mechanisms on the following 5 data sets.
The first two data sets are real, and the other three ones are synthetic.
\begin{table}[t!]\vspace{-3ex}
	\caption{\small Experimental Settings.} \label{tab:settings}
	\centering
		\resizebox{8.7cm}{!}{
			\begin{tabular}{l|l}
				{\bf \quad \quad Parameters} & {\bf \qquad \qquad \qquad Values} \\ \hline \hline
				the norm distance, $b$               & $\lfloor0.33\check{b}\rfloor$,  $\lfloor0.67\check{b}\rfloor$, $\mathbf{\check{b}}$, $\lfloor1.33\check{b}\rfloor$, $\lfloor1.67\check{b}\rfloor$\\\hline
				the discrete side length, $d$                   &  $1$, $2$, $3$, $4$, $\mathbf{5}$, $10$, $\underline{15}$, $20$\\\hline
				the privacy budget, $\epsilon$      & $0.7$, $1.4$, $2.1$, $2.8$, $\mathbf{3.5}$, $\underline{5}$, $6$, $7$, $8$, $9$ \\
				\hline
			\end{tabular}
		}\vspace{-1ex}
\end{table}

\emph{Chicago Crime}~\cite{ChicagoCrimes} (Crime): It is collected to monitor crime events in Chicago from January 1st to June 30th, 2022. It contains 105,453 data items, each representing a crime event. We extract events with a latitude range $[40^{\circ}, 42^{\circ}]$ and a longitude range $[-87.9^{\circ}, -87.54^{\circ}]$. Finally, we get $101,146$ items.

\emph{NYC Green Taxis}~\cite{NYCGreen2016} (NYC): It records green taxi order information in New York City in 2016. It contains $448,181$ order items. We only extract orders with a pickup location latitude and longitude within the ranges of $[40.55^{\circ}, 40.88^{\circ}]$ and $[-74.05^{\circ}, -73.73^{\circ}]$, respectively.
Finally, we get $446,110$ items.

The latitude and longitude of Chicago crime event locations and NYC green taxi pick-up locations are shown in Figure~\ref{subfig:ChicagoCrimesMarked} and~\ref{subfig:NYCGreenTaxiMarked}. We project the latitude and longitude onto a plane, which does not affect our experimental results.

To address the irregularity of these positions, we further extract three parts (marked as squares in Figure~\ref{subfig:ChicagoCrimesMarked} and~\ref{subfig:NYCGreenTaxiMarked}) of the two real data sets and estimate the distributions within each part. 
Table~\ref{dataset:parts} shows the number of data points in each area for Chicago Crimes and NYC Green Taxis. 
\revision{For the experiment on the full domain of the real data sets, please refer to Appendix~C.}

\emph{Normal(0, 0, 1, 1, 0.5)} (Normal): We generate $300,000$ 2-Dim data points in the plane, where each point follows a 2-Dim Gaussian distribution with $\mu_x = \mu_y = 0$, $\sigma_x^2 = \sigma_y^2 = 1$, and $\rho = 0.5$. The correlation between $x$ and $y$ is described by $\rho \in (-1, 1)$, where $-1 < \rho < 0$ indicates negative correlation, $0 < \rho < 1$ indicates positive correlation, and $\rho = 0$ indicates independence between $x$ and $y$. The points are all within the range $(-5,5)\times(-5,5)$. 
\revision{The range of these points is $[-4.44, 4.65]\times[-4.87, 4.58]$.}

\emph{Skew Zipf($\frac{1}{\ln{2}}$, 1,1)} (SZipf): We generate 100,000 2-Dim data points in the plane. Each dimension of each point follows a skew Zipf distribution with a CDF of $\frac{1/\ln{2}}{X+1}$. The points are limited to the range $[0,1)\times[0,1)$. We show one-tenth of the point distribution in Figure~\ref{subfig:ZipfDataset}.

\emph{Multi-center Normal} (MNormal): We generate $300,000$ 2-Dim data points in the plane. These points can be divided into three parts, each containing $100,000$ points. The parts follow a normal distribution with parameters \emph{Normal(0, 0, 1, 1, 0.5)}, \emph{Normal(0, 0, 1, 1, 0)}, and \emph{Normal(0, 0, 1, 1, -0.2)}, respectively. 
\revision{The range of these points is $[-4.25, 6.18]\times[-4.32, 6.44]$.}

\begin{figure}[t!]\centering
	\subfigure[][{\small Chicago Crimes}]{
		\scalebox{0.25}[0.25]{\includegraphics{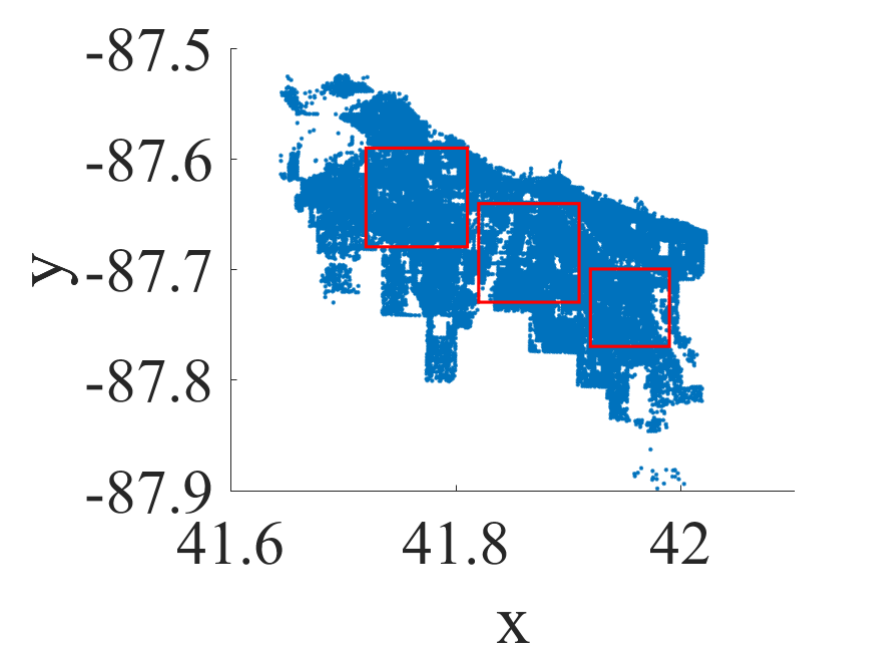}}
		\label{subfig:ChicagoCrimesMarked}}
	\subfigure[][{\small NYC Green Taxis}]{
		\scalebox{0.25}[0.25]{\includegraphics{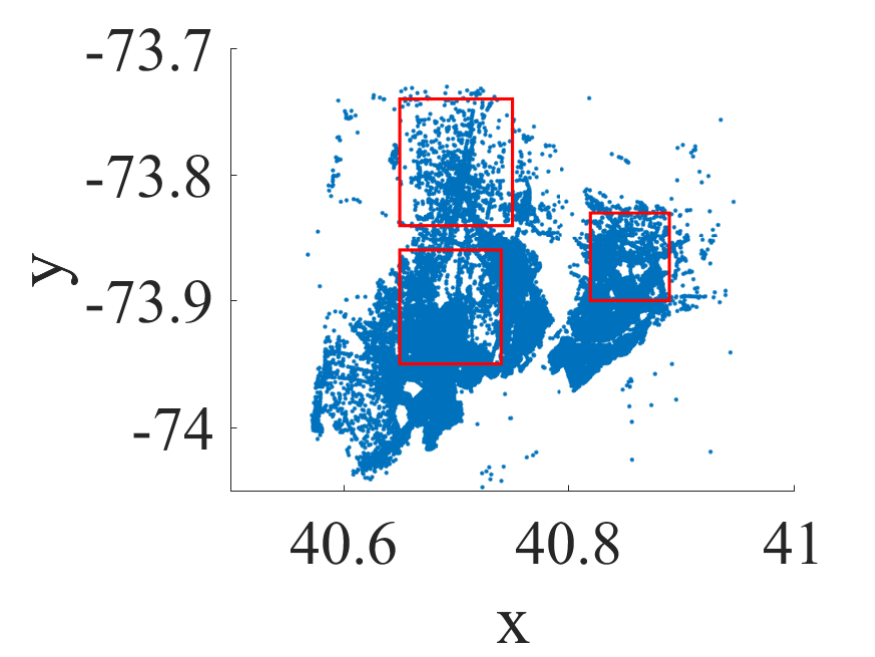}}
		\label{subfig:NYCGreenTaxiMarked}}
	\subfigure[][{\small Normal}]{
		\scalebox{0.18}[0.18]{\includegraphics{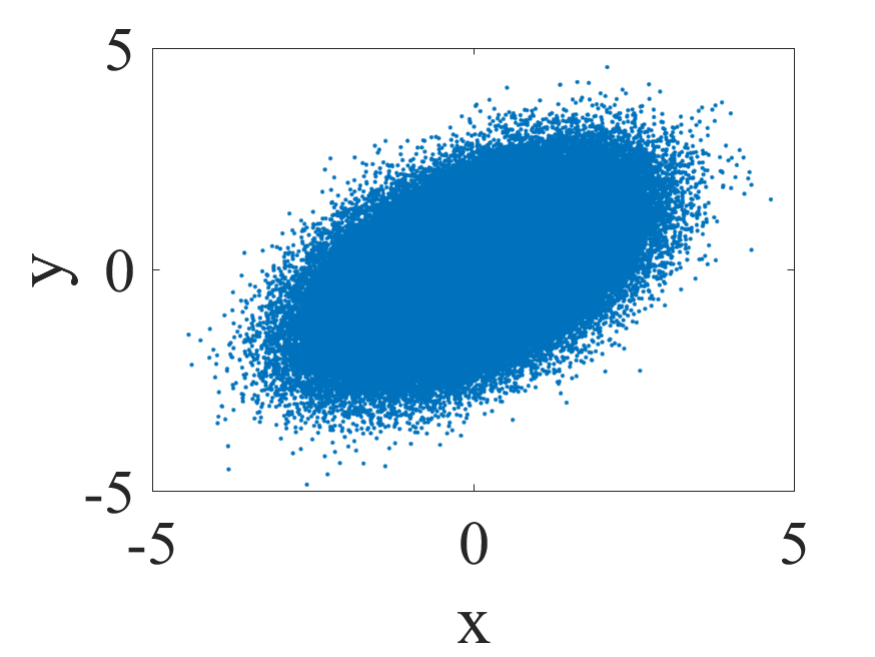}}
		\label{subfig:NormalDataset}}\hfill
	\subfigure[][{\small SZipf}]{
		\scalebox{0.18}[0.18]{\includegraphics{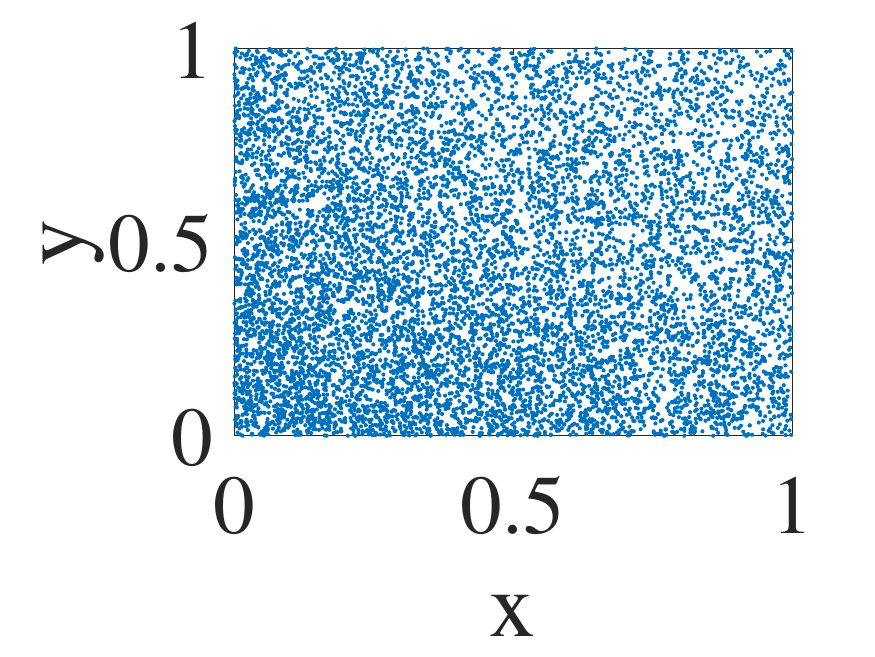}}
		\label{subfig:ZipfDataset}}\hfill
	\subfigure[][{\small MNormal}]{
		\scalebox{0.18}[0.18]{\includegraphics{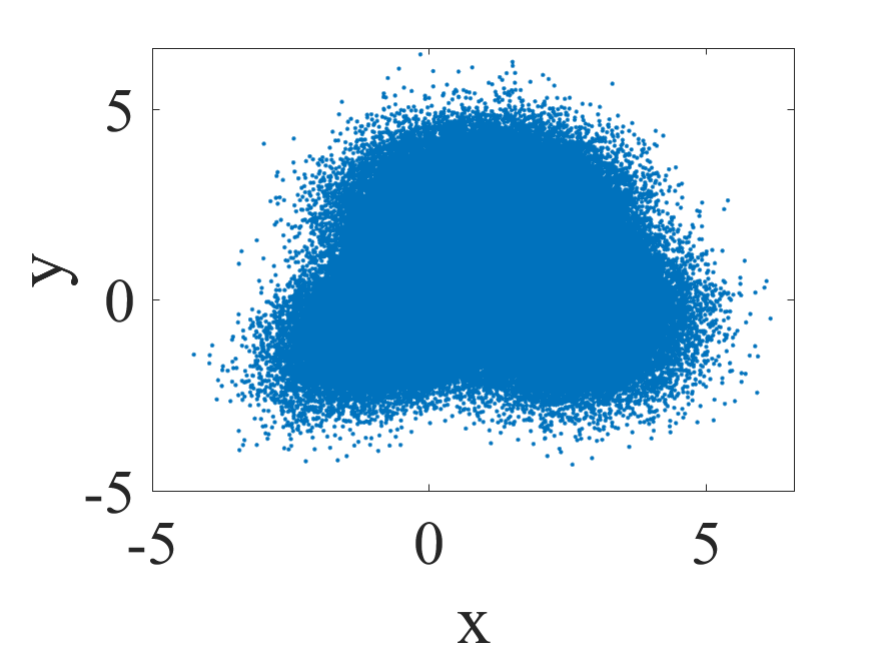}}
		\label{subfig:NormalDatasetMultipleCenters}}
	\caption{\small Data sets.}
	\label{fig:Data_sets_marked}
\end{figure}

\noindent
\textbf{Parameter Settings.}
We vary the norm distance $b$ from $0.33\check{b}$ to $1.67\check{b}$, where $\check{b}$ is the best choice of $b$.
We define the discrete side length $d$ as $L/g$, where $L$ is the side length of the input data set area, $g$ is the side length of a grid cell.
We change $d$ from $1$ to $20$ and $\epsilon$ from $0.7$ to $9$. The parameter settings are shown in Table~\ref{tab:settings} with default values marked as bold or underlined.

We conduct our experiment in Java on an Intel(R) Xeon(R) Silver 4210R CPU @ 2.4GHz with 128 GB RAM. We run our experiment $10$ times and use the average result as our final result.

\subsection{Mechanisms and Measures}\vspace{-1ex}
We compare our mechanisms \solutionD{} and \solutionB{} with \solutionCMPmTotalName{} (\solutionCMPm{})~\cite{DBLP:journals/pvldb/Yang0L0S20} and \solutionCMPgTotalName{} (\solutionCMPg{})~\cite{DBLP:conf/infocom/WangNWXYH17}.
We also compare our \solutionB{} with its version without shrinkage (i.e., \solutionBNSTotalName{}, \solutionBNS{}).
\revision{
	Furthermore, we evaluate our mechanism against recent research on private trajectory estimation: \solutionCMPTrATotalName{} (\solutionCMPTrA{})~\cite{DBLP:journals/pvldb/DuHZFCZG23} and \solutionCMPTrBTotalName{} (\solutionCMPTrB{})~\cite{DBLP:journals/pvldb/Zhang000H23}. 
	For detailed comparisons with the trajectory mechanisms, please refer to Appendix~D.}

\revision{Actually, it is hard to make \solutionB{} and \solutionCMPg{} comparable because \solutionB{} and \solutionCMPg{} are based on different privacy definition (\solutionB{} is based on LDP while \solutionCMPg{} is based on Geo-I).
	However, the definitions of both LDP and Geo-I are based on \textit{privacy loss}~\cite{DBLP:conf/icalp/Dwork06} which is a more fundamental privacy definition.
	Thus, given the same input domain, we can set the same privacy loss in \solutionB{} and \solutionCMPg{} and compare the utilities between these two mechanisms.
}

\revision{
	Givena a fixed privacy budget $\epsilon$, \solutionB{} achieves $\epsilon$-LDP which means for any $\vectorfont{v}\in\algvar{D}$, the privacy loss of randomizing $\vectorfont{v}$ to $\tilde{\vectorfont{v}}\in\tilde{\algvar{D}}$ is $\epsilon$. 
	However, \solutionCMPg{} achieves $\epsilon$-Geo-I which means for any $\vectorfont{v}\in\algvar{D}$, the privacy loss of randomizing $\vectorfont{v}$ to $\tilde{\vectorfont{v}}\in\tilde{\algvar{D}}$ is $\epsilon\cdot dis(\vectorfont{v},\tilde{\vectorfont{v}})$ where
	$dis(\vectorfont{v},\tilde{\vectorfont{v}})$ is the Euclidean distance between $\vectorfont{v}$ and $\tilde{\vectorfont{v}}$ (i.e., $dis(\vectorfont{v},\tilde{\vectorfont{v}})=\|\vectorfont{v}-\tilde{\vectorfont{v}}\|_2$). 
	We find if $dis(\vectorfont{v},\tilde{\vectorfont{v}})<1$, \solutionCMPg{} provides higher level privacy protection than \solutionB{}, if $dis(\vectorfont{v},\tilde{\vectorfont{v}})>1$, \solutionB{} provides higher level privacy protection than \solutionCMPg{}. 
}

Next, we introduce the definition of an enhanced loss privacy called Local Privacy (LP)~\cite{DBLP:conf/ccs/ShokriTTHB12} to make \solutionB{} and \solutionCMPg{} comparable on both utility and privacy.
Let $\mathcal{I}$ be the input domain, and $\mathcal{T}$ be the output domain. Let $\mathcal{\hat{I}}$ be the inferred domain of $\mathcal{I}$. Based on the unbiased results of \solutionB{} and \solutionCMPg{}, we have $\mathcal{\hat{I}}=\mathcal{I}$. Thus, the local privacy is defined as:

\vspace{-1ex}
\begin{equation}\label{Def_LP}
	{\scriptsize
		\begin{aligned}
			LP  &= \sum_{i'\in\mathcal{T}}LP_{\mathcal{I}}(i') = \sum_{i,\hat{i}\in\mathcal{I}}LP_{\mathcal{T}}(i,\hat{i})\\
			&= \sum_{i,\hat{i}\in\mathcal{I},i'\in\mathcal{T}}\Pr(i)\Pr(i'|i)\Pr(\hat{i}|i')d_p(\hat{i},i)
		\end{aligned}
	}
\end{equation}
In Equation~\eqref{Def_LP}, $\Pr(i)$ is the probability of being at location $i$ when accessing the location-based service. $\Pr(i'|i)$ is the location obfuscation function implemented by privacy mechanisms, which is defined as the probability of replacing $i$ with $i'$. $\Pr(\hat{i}|i')$ is the adversary attack function, which is defined as the probability of estimating $\hat{i}$ as the user's actual location if $i'$ is observed. $d_p(\hat{i},i)$ is the privacy of the user at location $i$, given that the adversary's estimation is $\hat{i}$. This is defined as the distance between  $\hat{i}$ and $i$, using 2-norm distance.

Suppose the truthful points obey a uniform distribution. Then, we have $\Pr(i)=\frac{1}{n}$, where $n$ is the number of truthful locations. According to the unbiased estimation for LDP, we have $\Pr(\hat{i})=\Pr(i)=\frac{1}{n}$ and $\Pr(i'|i)=\Pr(i'|\hat{i})$ when $i=\hat{i}$. Therefore, for \solutionB{} and \solutionCMPg{}, we have:
\begin{equation}\label{Def_Sub_LP_DP}
	{\scriptsize
		\begin{aligned}
			LP_{\mathcal{I}}(i')  &= \frac{1}{n}\sum_{i,\hat{i}\in\mathcal{I}}\Pr(i'|i)\Pr(\hat{i}|i')d_p(\hat{i},i)\\
			&= \frac{1}{n}\sum_{i,\hat{i}\in\mathcal{I}}\frac{\Pr(i'|i)\cdot\Pr(i'|\hat{i})\Pr(\hat{i})\cdot d_p(\hat{i},i) }{\sum_{\hat{i}_j\in\mathcal{I}}\Pr(i'|\hat{i}_j)\Pr(\hat{i}_j)}\\
			&= \frac{1}{n}\sum_{i,\hat{i}\in\mathcal{I}}\frac{\Pr(i'|i)\Pr(i'|\hat{i})d_p(\hat{i},i)}{\sum_{\hat{i}_j\in\mathcal{I}}\Pr(i'|\hat{i}_j)}\\
			&= \frac{1}{n\sum_{\hat{i}_j\in\mathcal{I}}\Pr(i'|\hat{i}_j)}\sum_{i,\hat{i}\in\mathcal{I}}\Pr(i'|i)\Pr(i'|\hat{i})d_p(\hat{i},i)
		\end{aligned}
	}
\end{equation}
To calculate the value of $LP$ for \solutionB{}, traverse all $i'\in\mathcal{I'}$, where $\mathcal{I'}$ is the output domain of \solutionB{}. For \solutionCMPg{}, calculate the value of $LP$ by traversing all $i'\in\mathcal{S}_k$, where $\mathcal{S}_k$ is the $k$-subset domain of $\mathcal{I}$.
In our experiment, we set $\epsilon$ in \solutionB{} as the values in Table~\ref{tab:settings} and calculate the corresponding $\epsilon'$ in \solutionCMPg{} with their local privacy equal.
We compare the accuracy of the mechanisms above, described by the 2-Dim Wasserstein distance.
Since there is no closed-form solution for high-dimensional Wasserstein distance~\cite{panaretos2019statistical}, we calculate the 2-Dim Wasserstein distance using Linear Programming.

Note that in our models, all distributions are finite and the input and output domains are divided into grids. Therefore, in our models, the Wasserstein distance involves multidimensional optimization in finite spaces. We can formalize this as follows:

Suppose $\mathcal{D}=\{X_1,...,X_m\}$ and $\tilde{\mathcal{D}}=\{Y_1,...,Y_n\}$.
Let $M=\{\|X_i-Y_j\|_p^p\}_{i\in[m],j\in[n]}$ be the matrix where each element is the $p$ norm to the $p$ for the pair of $X$ and $Y$.
Let $R=\{\Pr[X_i,Y_j]\}_{i\in[m],j\in[n]}$ be the matrix where each element is the joint probability of $X_i$ and $Y_j$.
Then,
\begin{equation}\label{discrete_finite_wasserstein}
	{\scriptsize
		\begin{aligned}
			W_2^p(\mathcal{D},\tilde{\mathcal{D}})=\textrm{min}\;\; & \|M\bigodot R\|_F &\\
			s.t.\;\; & \textrm{sum}_r(R)=\Pr[\mathcal{D}] \\
			& \textrm{sum}_c(R)=\Pr[\tilde{\mathcal{D}}] \\
			& R_{i,j}\geq 0;\forall i\in[m]\;\textrm{and}\;j\in[n]
		\end{aligned}
	}
\end{equation}
where $\bigodot$ is the Hadmard production, $\|\cdot\|_F$ is the Frobenius norm.

\begin{figure}[t!]
	\centering
	\scalebox{0.63}[0.63]{\includegraphics{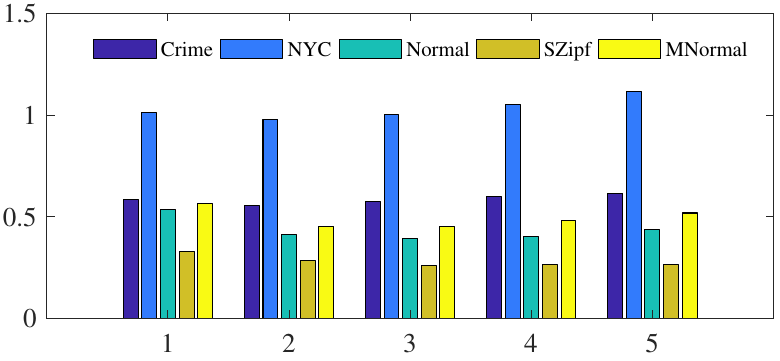}}
	\caption{Wasserstein distances with $b$ varied.}\label{fig:alter_b}
\end{figure}

\begin{figure*}[t!]\centering
	\subfigure{
		\scalebox{0.35}[0.35]{\includegraphics{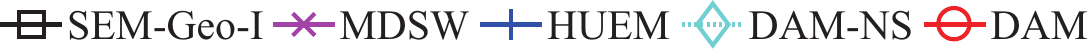}}}\hfill\\
	\addtocounter{subfigure}{-1}
	\subfigure[][{\small Crime}]{
		\scalebox{0.24}[0.24]{\includegraphics{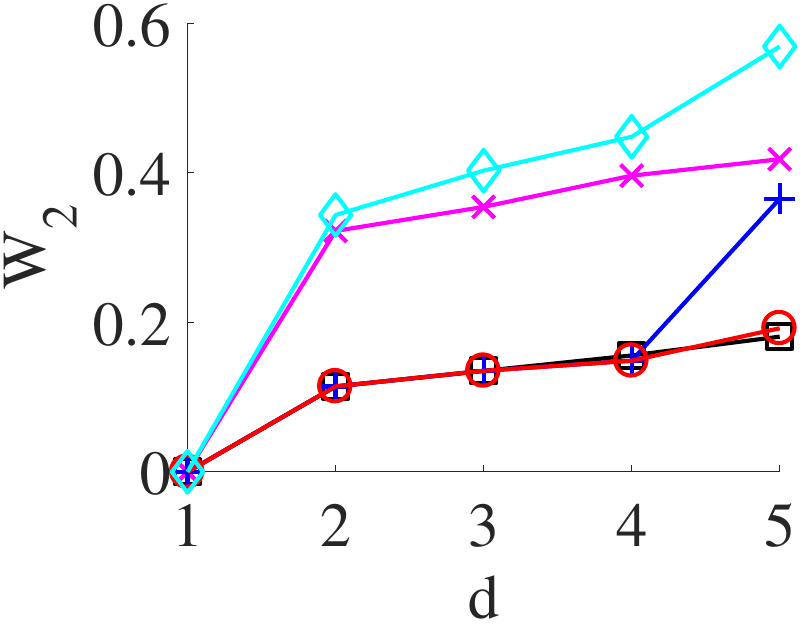}}
		\label{subfig:alter_d_crime}}\hfill	
	\subfigure[][{\small NYC}]{
		\scalebox{0.24}[0.24]{\includegraphics{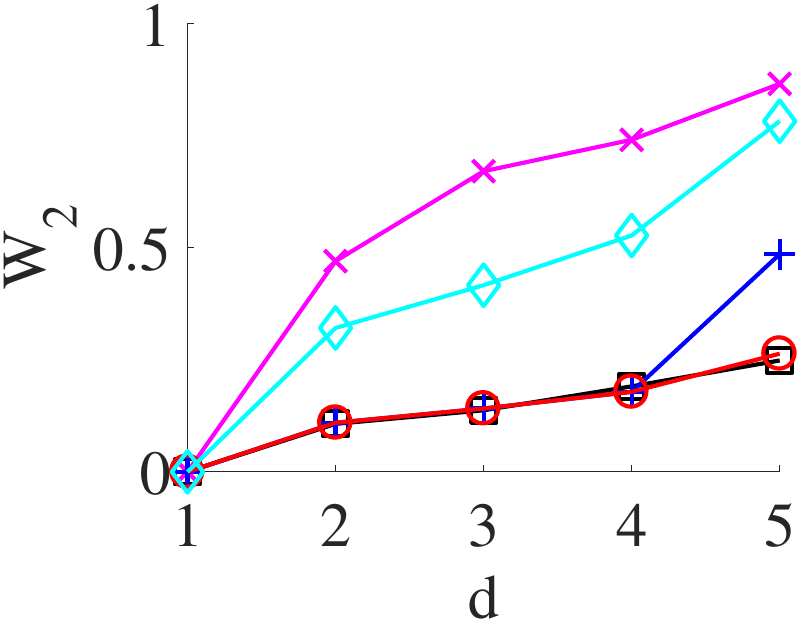}}
		\label{subfig:alter_d_nyc}}\hfill
	\subfigure[][{\small Normal}]{
		\scalebox{0.24}[0.24]{\includegraphics{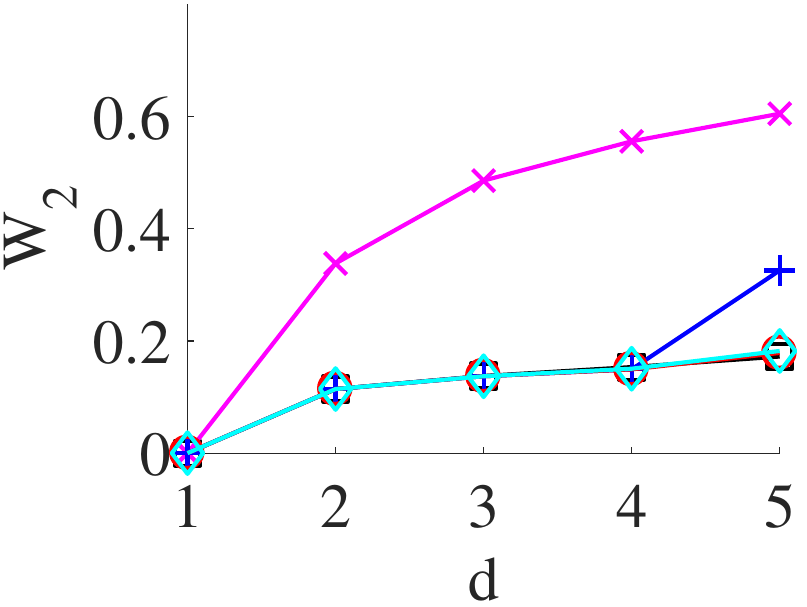}}
		\label{subfig:alter_d_normal}}\hfill 
	\subfigure[][{\small SZipf}]{
		\scalebox{0.24}[0.24]{\includegraphics{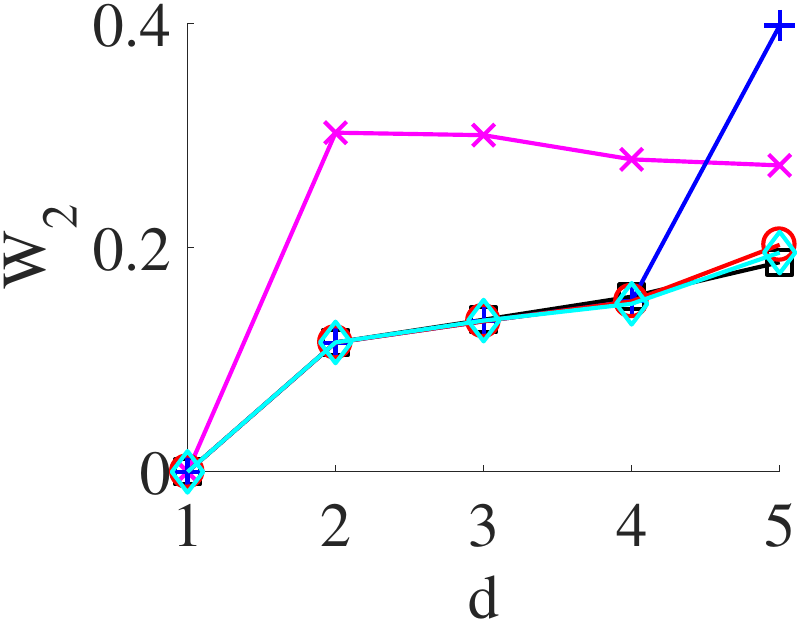}}
		\label{subfig:alter_d_zipf}}\hfill
	\subfigure[][{\small MNormal}]{
		\scalebox{0.24}[0.24]{\includegraphics{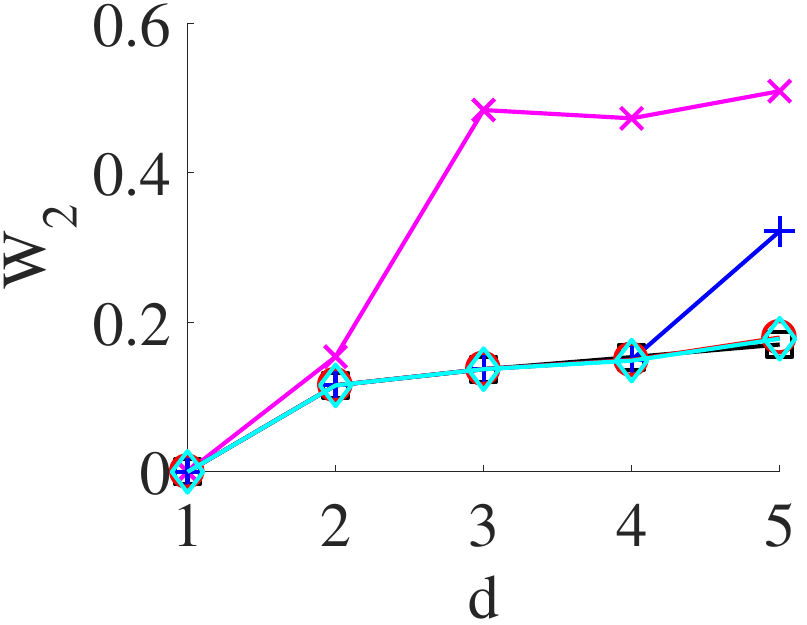}}
		\label{subfig:alter_d_normal_multiple_centers}}\newline\vspace{-1ex}
	
	\subfigure[][{\small Crime}]{
		\scalebox{0.24}[0.24]{\includegraphics{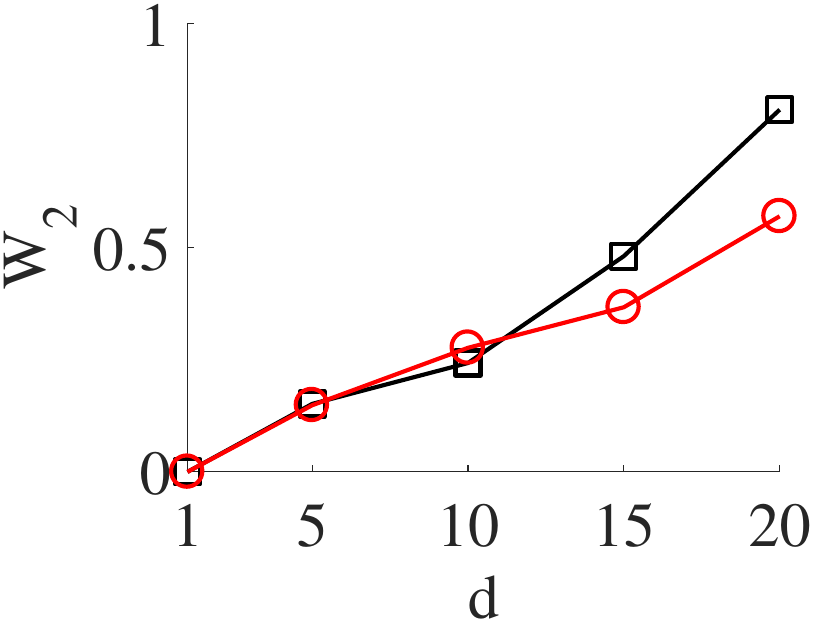}}
		\label{subfig:alter_g_crime_extended}}\hfill 
	\subfigure[][{\small NYC}]{
		\scalebox{0.24}[0.24]{\includegraphics{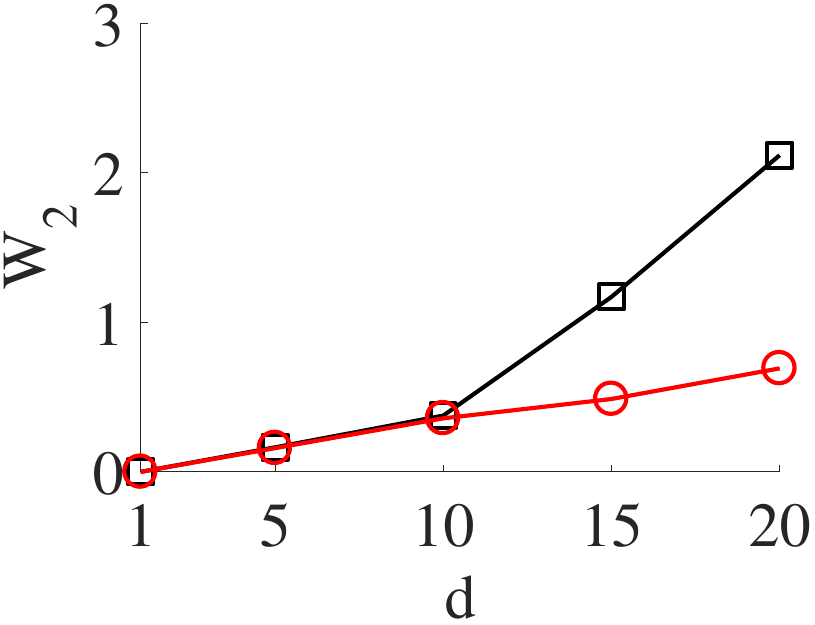}}
		\label{subfig:alter_g_nyc_extended}}\hfill 
	\subfigure[][{\small Normal}]{
		\scalebox{0.24}[0.24]{\includegraphics{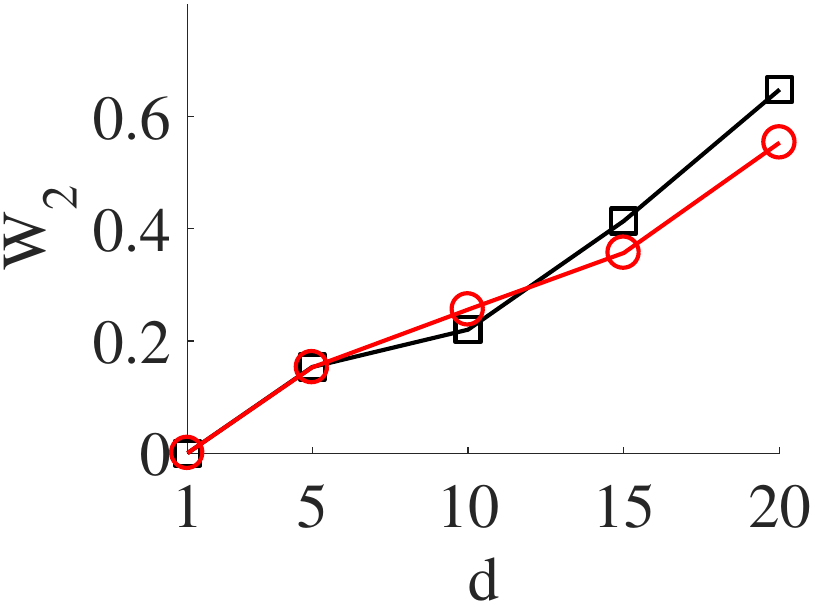}}
		\label{subfig:alter_g_normal_extended}}\hfill
	\subfigure[][{\small SZipf}]{
		\scalebox{0.24}[0.24]{\includegraphics{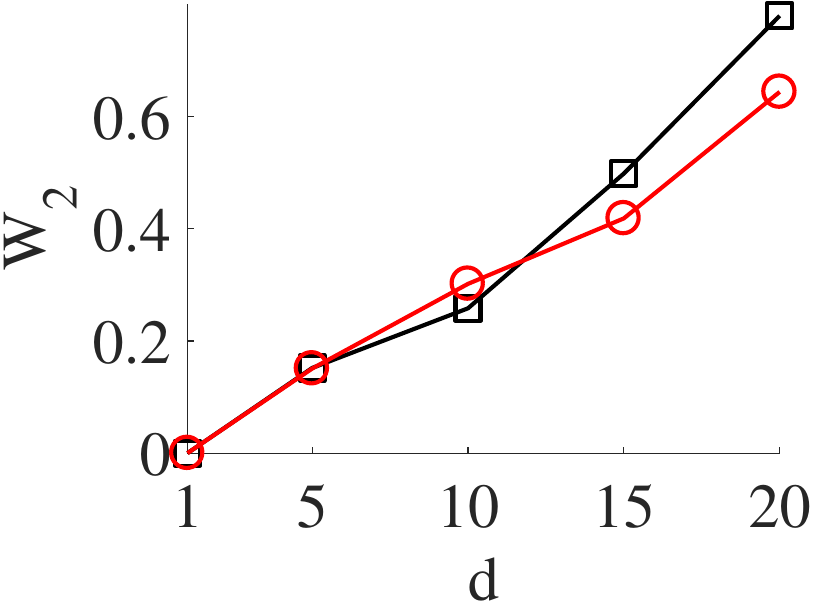}}
		\label{subfig:alter_g_zipf_extended}}\hfill 
	\subfigure[][{\small MNormal}]{
		\scalebox{0.24}[0.24]{\includegraphics{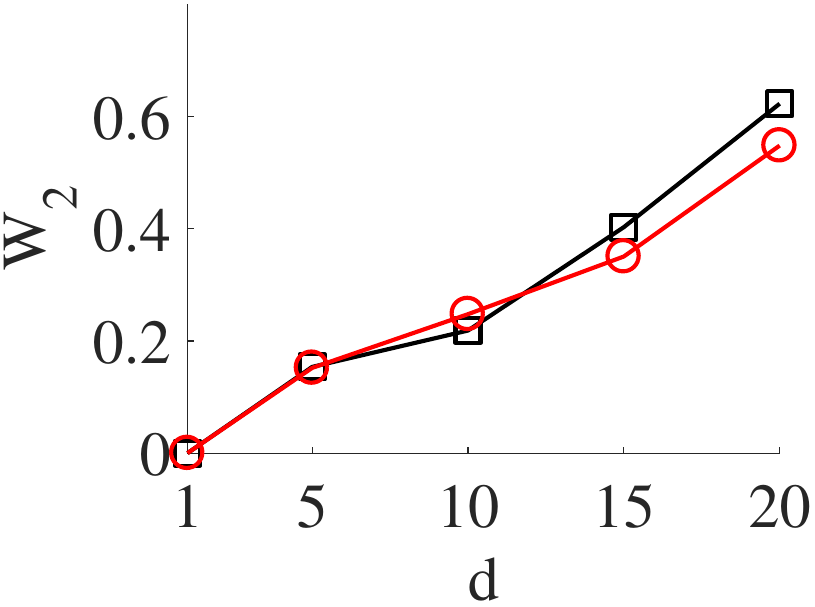}}
		\label{subfig:alter_g_normal_multiple_centers_extended}}\newline \vspace{-1ex}
	
	\subfigure[][{\small Crime}]{
		\scalebox{0.24}[0.24]{\includegraphics{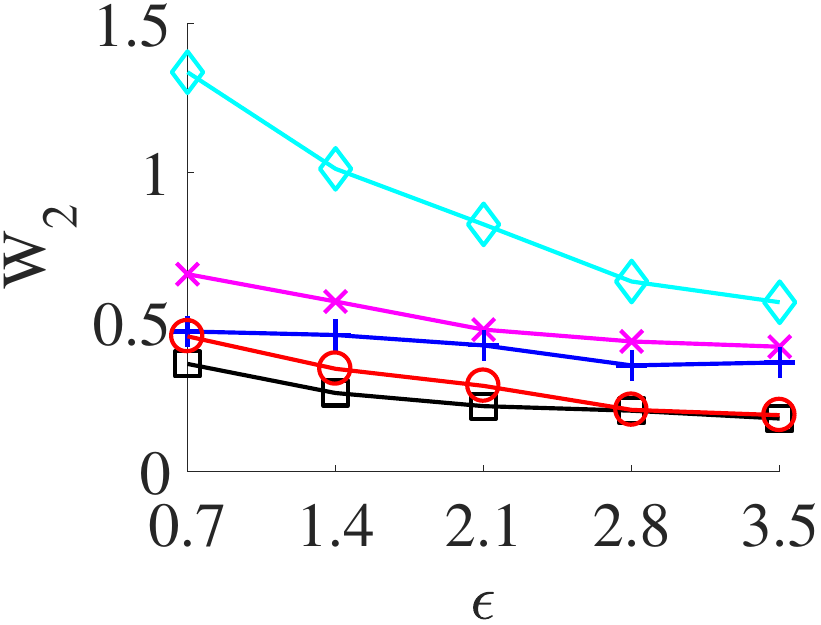}}
		\label{subfig:alter_e_crime}}\hfill 
	\subfigure[][{\small NYC}]{
		\scalebox{0.24}[0.24]{\includegraphics{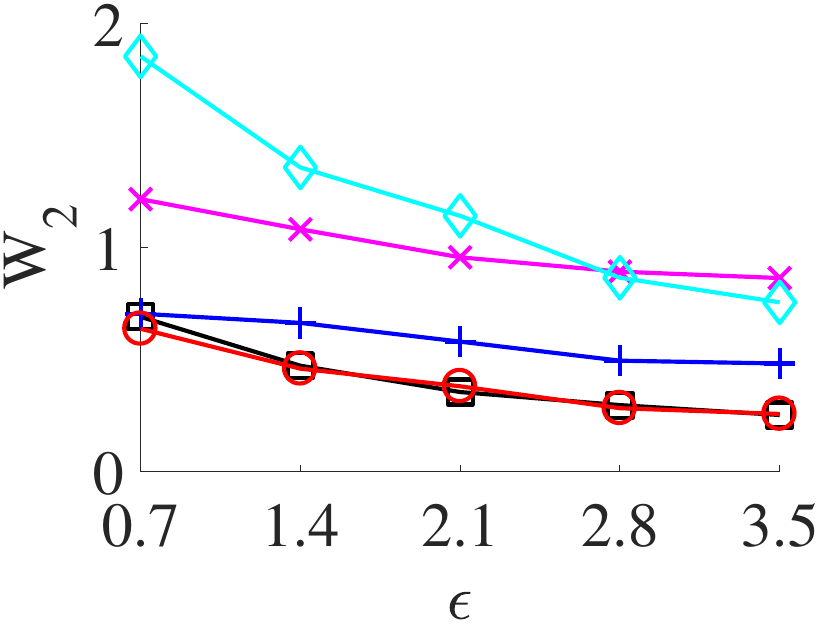}}
		\label{subfig:alter_e_nyc}}\vspace{-2ex}\hfill
	\subfigure[][{\small Normal}]{
		\scalebox{0.24}[0.24]{\includegraphics{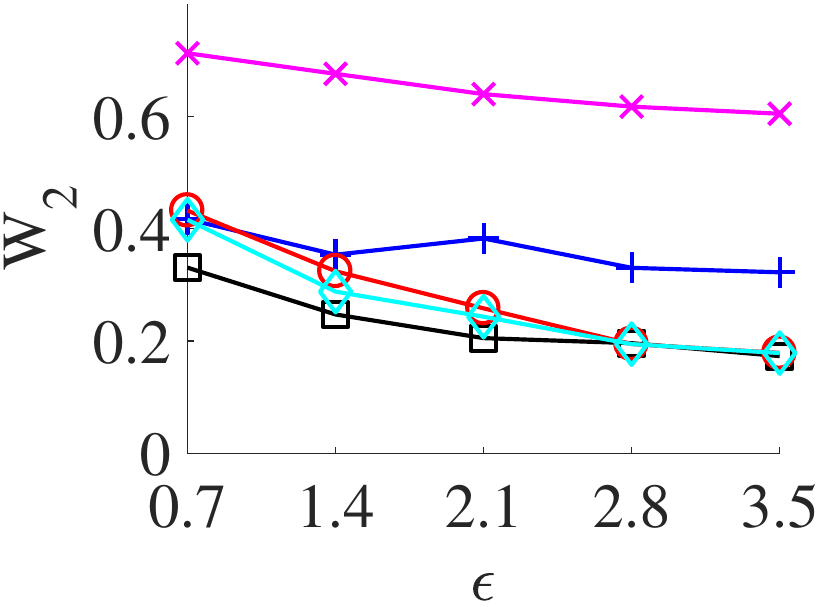}}
		\label{subfig:alter_e_normal}}\hfill
	\subfigure[][{\small SZipf}]{
		\scalebox{0.24}[0.24]{\includegraphics{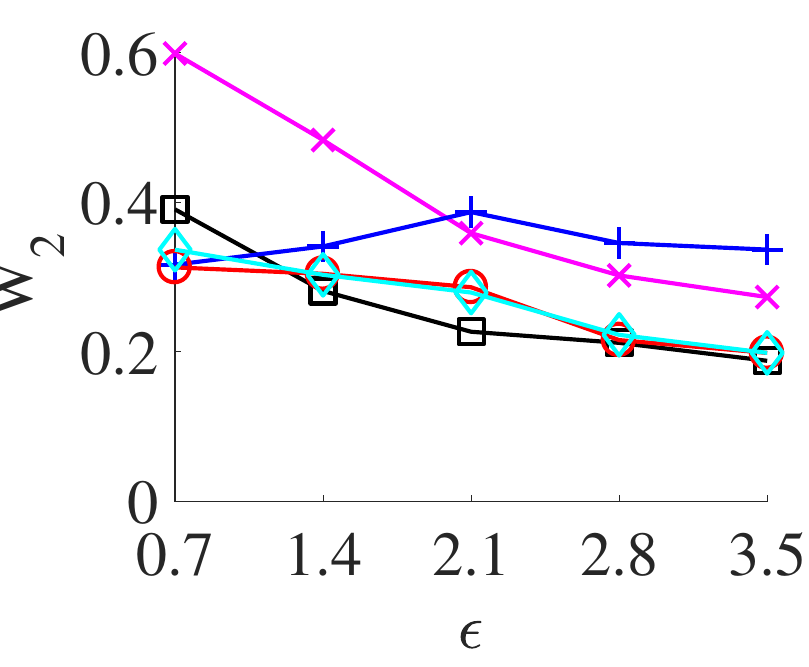}}
		\label{subfig:alter_e_zipf}}\hfill 
	\subfigure[][{\small MNormal}]{
		\scalebox{0.24}[0.24]{\includegraphics{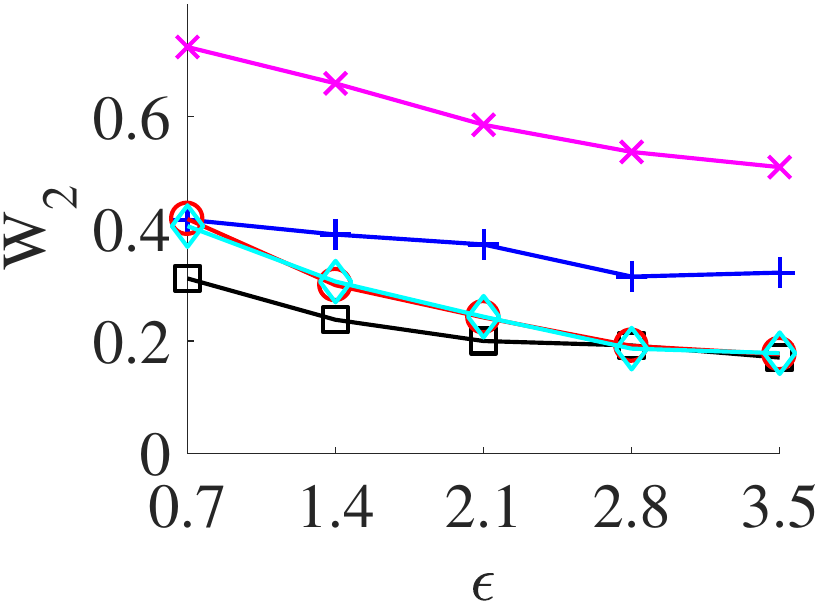}}
		\label{subfig:alter_e_normal_multiple_centers}}\newline \vspace{-1ex}

	\subfigure[][{\small Crime}]{
		\scalebox{0.24}[0.24]{\includegraphics{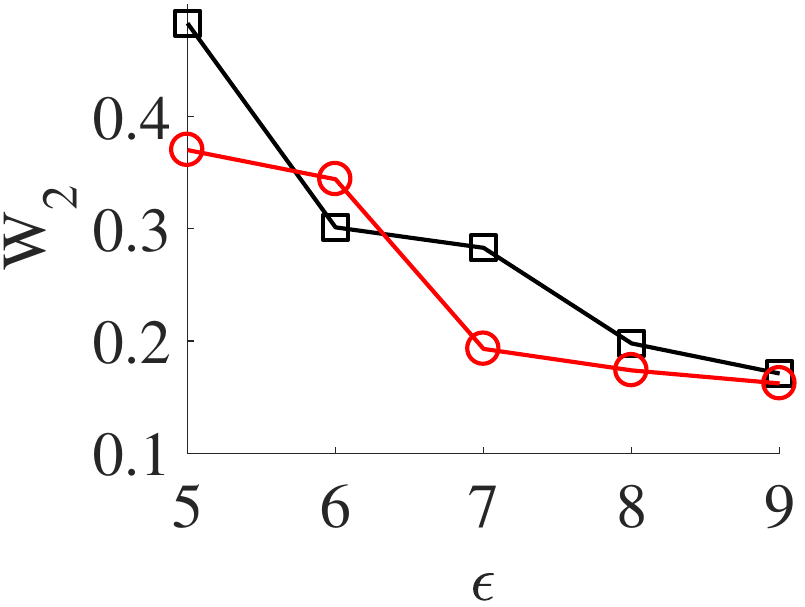}}
		\label{subfig:alter_e_crime_extended}}\hfill 
	\subfigure[][{\small NYC}]{\hspace{-2.1mm}
		\scalebox{0.24}[0.24]{\includegraphics{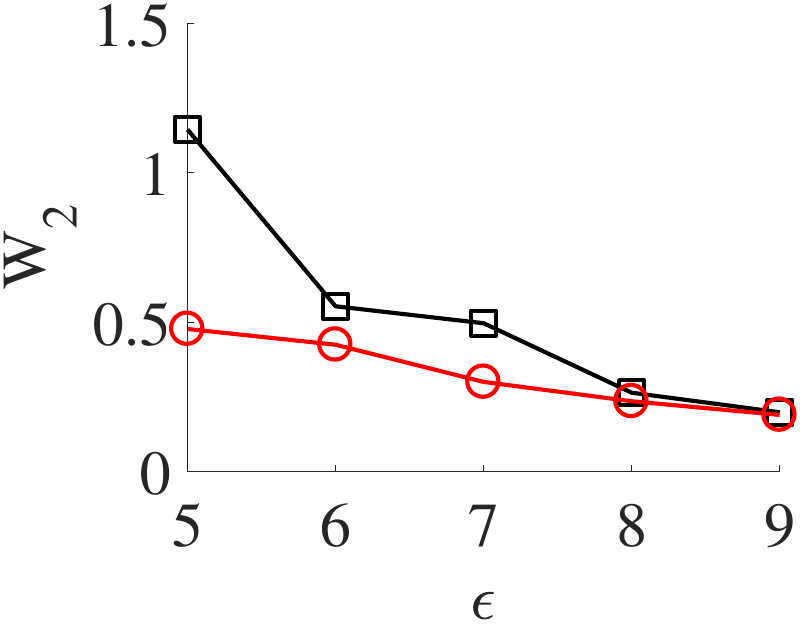}}
		\label{subfig:alter_e_nyc_extended}}\hfill 
	\subfigure[][{\small Normal}]{
		\scalebox{0.24}[0.24]{\includegraphics{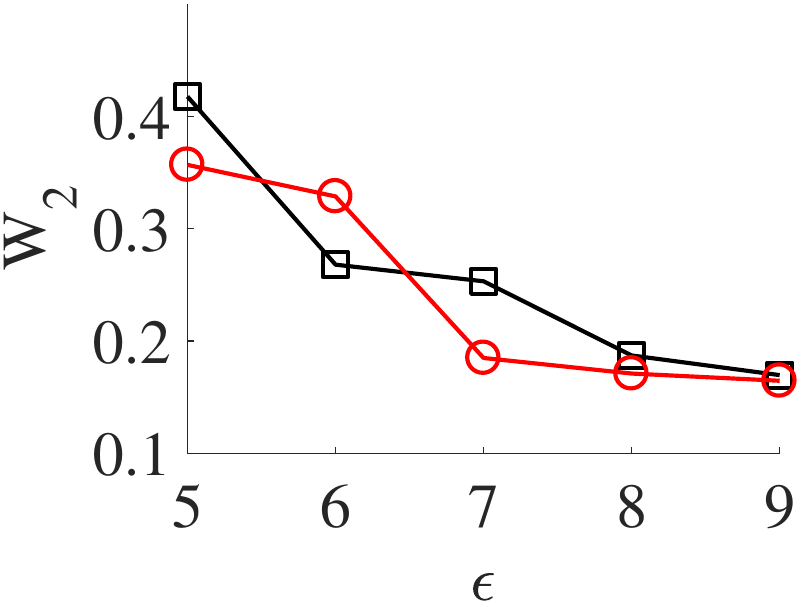}}
		\label{subfig:alter_e_normal_extended}} \hfill 
	\subfigure[][{\small SZipf}]{
		\scalebox{0.24}[0.24]{\includegraphics{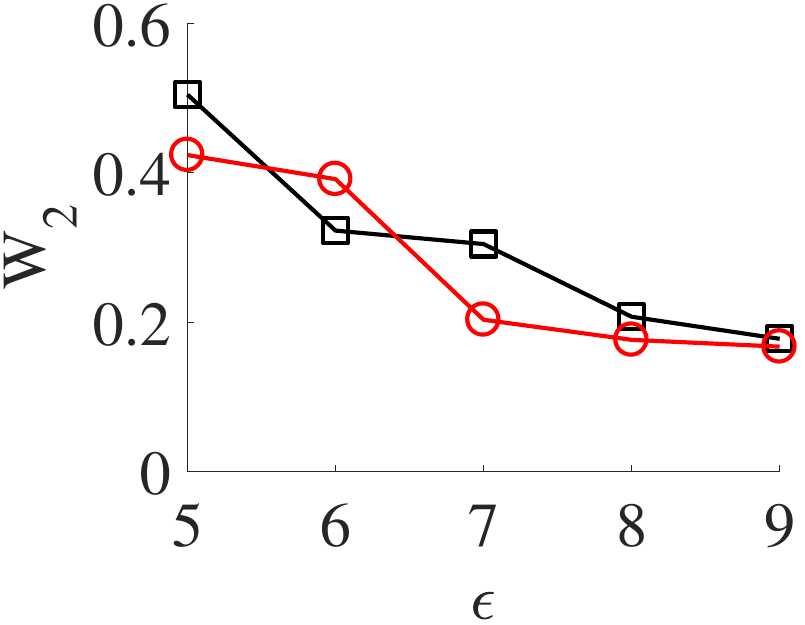}}
		\label{subfig:alter_e_zipf_extended}} \hfill 
	\subfigure[][{\small MNormal}]{
		\scalebox{0.24}[0.24]{\includegraphics{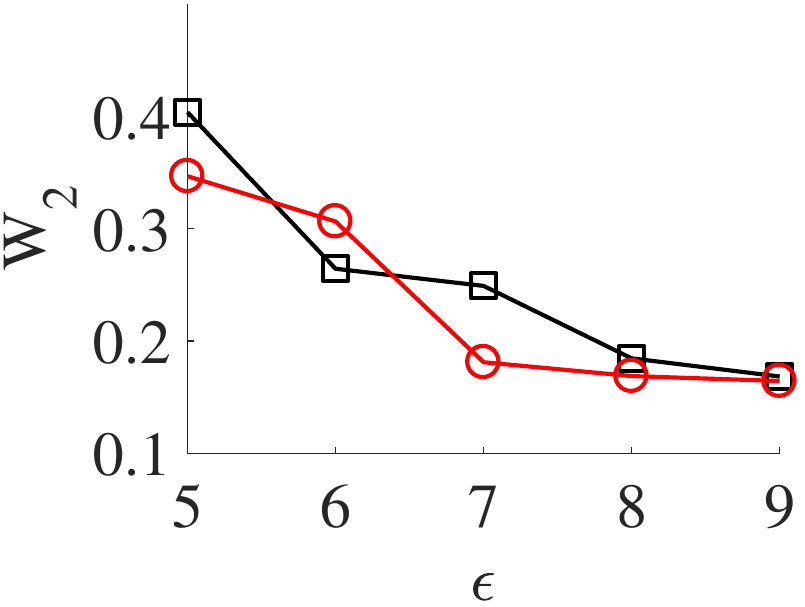}}
		\label{subfig:alter_e_normal_multiple_centers_extended}}\hfill
	\caption{\small Wasserstein distances with $d$ or $\epsilon$ varied.}\vspace{-2ex}
	\label{fig:alter_d_e}
\end{figure*}

\subsection{Experimental Results}\label{exp:result}
We compared the estimation results of methods for distribution estimation. We used the 2-norm Wasserstein distance $W_2=\sqrt{W_2^2}$ between the recovered and actual density distributions in a discrete situation. We compared the results on different data sets with varying values for the norm distance $b$, discrete side length $d$, and privacy budget $\epsilon$. For the divided data sets (Chicago Crime and NYC Green Taxis), we used the mean value of each part's $W_2$ as the estimation results.
\subsubsection{Norm Distance $b$}
Figure~\ref{fig:alter_b} shows the impact of the norm distance $b$ on $W_2$. The value of $b$ varies from $0.33\check{b}$ to $1.66\check{b}$, where $\check{b}$ is the optimal value of $b$ in a discrete situation. We set the default discrete side length $d=15$ and the default privacy budget $\epsilon=3.5$. In this case, the optimal norm distance $\check{b}$ is approximately equal to 3.
We can see that, in both the real and synthetic data sets, $W_2$ first decreases and then increases. When $b$ is approximately equal to $\check{b}$, $W_2$ achieves its minimum value, which is consistent with our analysis in section~\ref{subsection:b_analysis}. However, $W_2$ is not minimal when $b=\check{b}$ in some data sets due to the error from grid division.

\subsubsection{Discrete side length $d$}\label{exp:sub_side_d}
As the grid cells' side length $g$ decreases, $d$ increases with fixed $L$.
\revision{Here we vary $d$ from $1$ to $5$.}
As shown in Figure~\ref{subfig:alter_d_crime} to~\ref{subfig:alter_d_normal_multiple_centers}, $W_2$ increases with the increase of $d$ in most mechanisms except for \solutionCMPm{} on data set SZipf. This is because, as $d$ increases, the number of grid cells becomes larger, and the gap between the recovered and actual density distributions widens.
\revision{\solutionD{} is better than \solutionCMPm{} in most cases and \solutionB{} is always better than \solutionCMPm{}.
	It is because \solutionD{} and \solutionB{} retain the ordinal relationship of $x$-coordinate, $y$-coordinate and $(x,y)$-union among all points, however, \solutionCMPm{} only retains ordinal relationship of $x$-coordinate and $y$-coordinate.
	On the other hand, it indicates that considering the relationship between each dimension is useful. 
}
Additionally, our \solutionB{} is superior to \solutionD{}, which demonstrates its effectiveness in $2$-Dim area mechanisms. Moreover, \solutionB{} outperforms \solutionBNS{} in real data sets. This is because both data sets are road network data sets where the shrunken method has more advantages over the non-shrunken method.

\revision{The difference between \solutionCMPg{} and our \solutionB{} is small. That is because when $d$ is small, the side length of a grid cell is large, and the shape of discrete \solutionB{} is very different from that of continue \solutionB{}. This makes the discrete \solutionB{} performs worse than \solutionCMPg{}.}
We further compare these two mechanism under larger $d$.
However, as $d$ becomes larger, it becomes more difficult to calculate the Wasserstein distance within an acceptable time. Therefore, we use Sinkhorn's algorithm~\cite{DBLP:conf/nips/Cuturi13} to approximately calculate the Wasserstein distance.
We vary $d$ from $1$ to $20$ and set $\epsilon$ as $5$ to further compare \solutionCMPg{} and our \solutionB{} from Figure~\ref{subfig:alter_g_crime_extended} to Figure~\ref{subfig:alter_g_normal_multiple_centers_extended}.
We can see, as $d$ increases, the Wasserstein distances of both \solutionCMPg{} and \solutionB{} also increase. Our \solutionB{} is better than \solutionCMPg{} when $d$ is larger.
\revision{This occurs because as $d$ increases, the discrete \solutionB{} gradually approaches to the continuous \solutionB{} in shape, and its error from grid diminishes.
	Consequently, the advantages of \solutionB{} become increasingly apparent.}

\subsubsection{Privacy Budget $\epsilon$}
The privacy budget $\epsilon$ not only affects the reported probability, but also influences the norm distance $b$. Figures~\ref{subfig:alter_e_crime} to~\ref{subfig:alter_e_normal_multiple_centers} show how the value of $W_2$ changes with the change in $\epsilon$. As $\epsilon$ increases, $W_2$ decreases slightly. This is because a large $\epsilon$ leads to a high probability report of the real data set, which makes the recovered density closed to the actual one. Our solution, \solutionB{}, is always better than \solutionCMPm{}. As $\epsilon$ increases, \solutionB{} achieves better estimation than \solutionD{}. In addition, 
\revision{\solutionCMPg{} slightly outperforms our \solutionB{} when $\epsilon$ is small.
	That occurs because a small $\epsilon$ causes the high probability domain to cover the input domain, making the differences between input cells less distinguishable. This will diminishes \solutionB{}'s advantage. 
	For \solutionCMPg{}, as $\epsilon$ decreases, its output domain space complexity grows by $n^{k}=O(n^{n/e^{\epsilon}})$ ($n=d^2$), exceeding our experiments' tolerance range for large $d$. To keep \solutionCMPg{} feasible, we must set $d$ to a small value when $\epsilon$ is small. However, a small $d$ further distorts the shape of discrete \solutionB{} (see the \textit{Discrete side length $d$} analysis in Subsection~\ref{exp:sub_side_d}), leading to its poor performance.} 
We further compare our \solutionB{} to \solutionCMPg{} under larger $\epsilon$ by Sinkhorn's algorithm~\cite{DBLP:conf/nips/Cuturi13} in Figure~\ref{subfig:alter_e_crime_extended} to~\ref{subfig:alter_e_normal_multiple_centers_extended}.
We set the $d$ as $15$ and vary $\epsilon$ from $5$ to $9$. In both \solutionB{} and \solutionCMPg{}, $W_2$ decreases as $\epsilon$ increases. As $\epsilon$ becomes larger, $W_2$ of the two mechanisms approach $0$, because a larger $\epsilon$ leads to higher accuracy for private distribution estimation. \solutionB{} outperforms \solutionCMPg{} when $\epsilon$ is large.

\section{Conclusion}
In this paper, we study 
\problemDefineTotalName{}.
We propose a general framework called \solutionGeneralTotalName{} (\solutionGeneralName{}) and a simple mechanism called \solutionDTotalName{} (\solutionD{}).
We further propose the optimal solution \solutionB{} among all \solutionGeneralName{}, leveraging the ordinal relationship between each data point to improve the accuracy of private distribution estimation.
Besides, we propose a shrinkage method to improve the estimation accuracy in the grid circumstance.
What's more, 
we compare our \solutionB{} with the state-of-the-art mechanisms to demonstrate that \solutionB{} can achieve the minimum Wasserstein distance among all mechanisms.

\section*{Acknowledgment}
Peng Cheng’s work is supported by the National Natural Science Foundation of China under Grant No. 62102149. Libin Zheng’s work is supported by  National Natural Science Foundation of China No. 62472455 and U22B2060. Xiang Lian’s work is supported by Natural Science Foundation (NSF CCF-2217104). Lei Chen’s work is partially supported by National Key Research and Development Program of China Grant No. 2023YFF0725100, National Science Foundation of China (NSFC) under Grant No. U22B2060, the Hong Kong RGC GRF Project 16213620, RIF Project R6020-19, AOE Project AoE/E-603/18, Theme-based project TRS T41-603/20R, CRF Project C2004-21G, Guangdong Province Science and Technology Plan Project 2023A0505030011, Hong Kong ITC ITF grants MHX/078/21 and PRP/004/22FX, Zhujiang scholar program 2021JC02X170, Microsoft Research Asia Collaborative Research Grant, HKUST-Webank joint research lab and HKUST(GZ)-Chuanglin Graph Data Joint Lab.
Xuemin Lin’s work is supported by NSFC U2241211 and U20B2046. Corresponding Author: Peng Cheng.

\bgroup\small
\bibliographystyle{IEEEtran}
\bibliography{reference}
\egroup

\appendix

\subsection{Grid Implement for \solutionD{}}\label{Imp:basic}
Our \solutionD{} can be regard as a union of $\hat{b}$ \solutionB{}.
Let $O_{(x,y),r}$ be the circle in position $(x,y)$ with radius $r$.
Given an integer $\hat{b}$, we can split it into $(b+1)$ parts by circle set $\{O_{(0,0),1},O_{(0,0),2},...,O_{(0,0),\hat{b}}\}$.
Figure~\ref{subfig:Appendix_grid_basic} shows the quarter parts of the division when $b=4$ and $d_1, d_2, d_3, d_4$ equal $1,2,3,4$ ,respectively.

For the part out of circle $O_{(0,0),\hat{b}}$, the reported probability of each cell is the same as that in discrete \solutionB{} with $b=\hat{b}$. For the part within circle $O_{(0,0),1}$, it is also the same as that in discrete \solutionB{} with $b=1$.
As for the middle parts, let $A_{\kappa}$ be the area $O_{(0,0),\kappa}\setminus O_{(0,0),\kappa-1}$, where $2\leq \kappa\leq \hat{b}$.
Each unit cell will be reported with probability $p_{d_{\kappa}}=q\cdot e^{(1-\frac{d_{\kappa}-1}{b})\epsilon}$.
For those cells split by the edge of $O_{(0,0),\kappa}$, each of them can be divided as two parts: shrank area $A_s$ and remain area $A_r$. Its reported probability is the weighted sum of these two parts, expressed as $p_{d_{\kappa-1}}\cdot A_s + p_{d_{\kappa}}\cdot A_r$.

\subsection{Proof of Theorems}
\subsubsection{Proof of Theorem~\ref{thrm:A_q}}\label{appendix:thrm_A_q}
\begin{proof}
	Let $S_{\mathcal{\hat{D}}}$ be the area size of $\mathcal{\hat{D}}$.
	Let $S_{p}$ be the area size of the pure high probability area $A_p$.
	Let $S_m$ be the area size of the mixed probability area $A_m$, which can be divided into the high probability area $A_{m,p}$ and low probability areas $A_{m,q}$ with area sizes $S_{m,p}$ and $S_{m,q}$.
	Then, we have $S_p+S_m+S_q=S_{\mathcal{\hat{D}}}$, $S_L=S_q+S_{m,q}$ and $S_H=S_p+S_{m,p}$.
	$S_q$ can be see as the increment from the area with $(S_p+S_m)$ area size centered at index $(0,0)$ to area $A_{\mathcal{\hat{D}}}$.
	We can divide the process of the increment $\delta$ into two sub-processes: \textit{moving up} ($\delta_u$) and \textit{moving right} ($\delta_r$).
	For the moving up, $\delta_u$ is $(2\hat{b}+1)(d-1)$ (shown in Figure~\ref{subfig:Moving_up}).
	For the moving right, $\delta_r$ is $(2\hat{b}+d)(d-1)$ (shown in Figure~\ref{subfig:Moving_right}).
	Therefore, $S_q=\delta=\delta_u+\delta_r=d^2+4\hat{b}d-4\hat{b}-1$.
\end{proof}

\begin{figure}[t!]\centering
		\subfigure[][{\small Division for \solutionD{}}]{
			\scalebox{0.43}[0.43]{\includegraphics{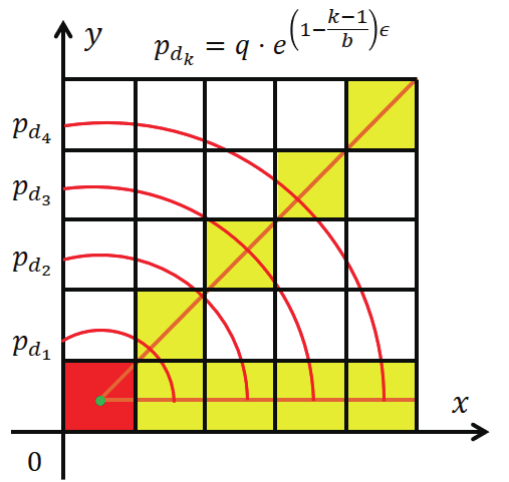}}
			\label{subfig:Appendix_grid_basic}}\hfill
		\subfigure[][{\small Division of $S_{\hat{b},(0,\frac{\pi}{4})}$}]{
			\scalebox{0.28}[0.28]{\includegraphics{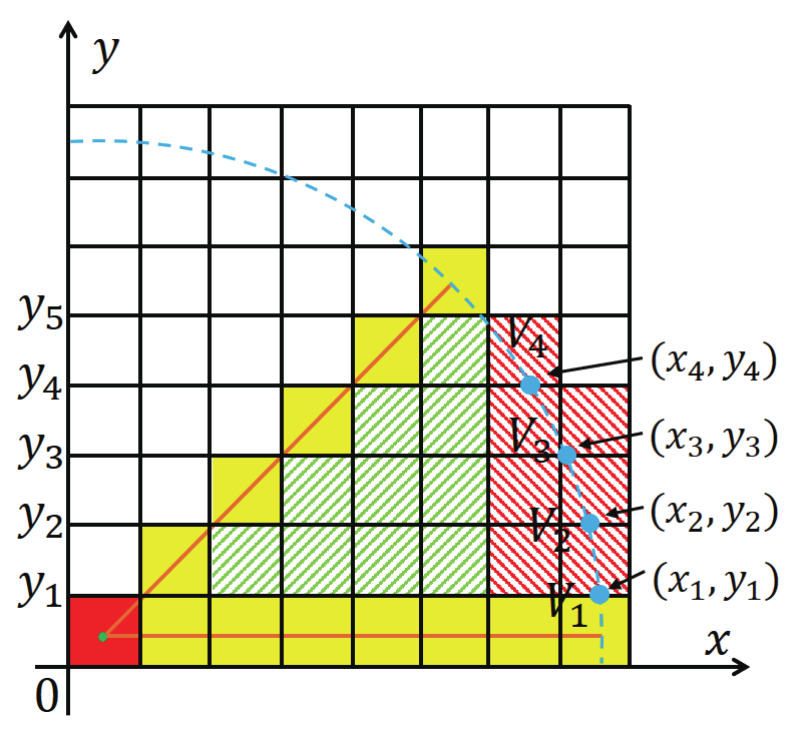}}
			\label{subfig:Appendix_so_count3}}
		\caption{\small Division for \solutionD{} and division of $S_{\hat{b},(0,\frac{\pi}{4})}$.}
		\label{fig:appendix_Approximate}
	\end{figure}
	
	\begin{figure}[t!]\centering\vspace{2ex}
		\subfigure[][{\small Moving up}]{
			\scalebox{0.37}[0.37]{\includegraphics{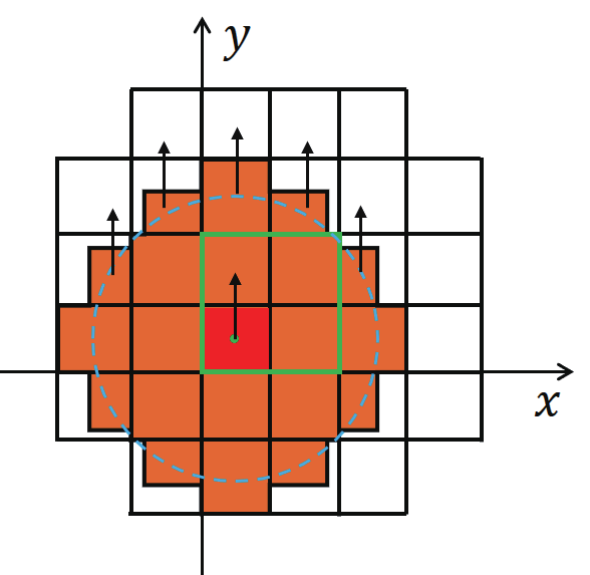}}
			\label{subfig:Moving_up}}\hfill
		\subfigure[][{\small Moving right}]{
			\scalebox{0.37}[0.37]{\includegraphics{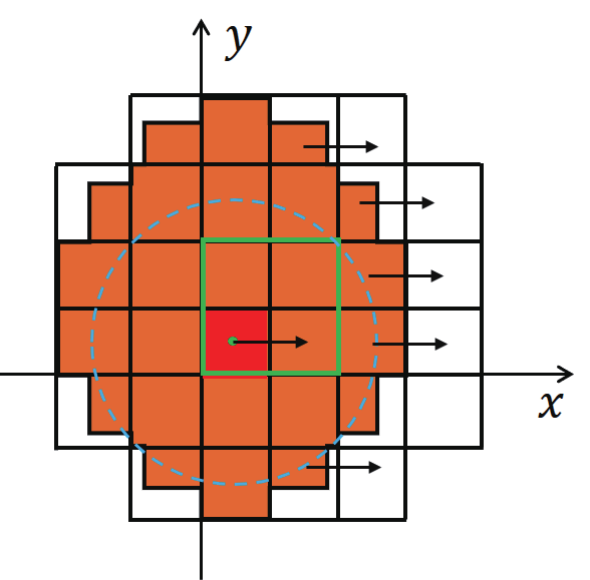}}
			\label{subfig:Moving_right}}
		\caption{\small Increment of $A_p\cup A_m$ in the process of getting $A_q$.}
		\label{fig:Moving}
	\end{figure}

	\subsubsection{Proof of Theorem~\ref{thrm:SOCount}}\label{appendix:thrm_SOCount}
	\begin{figure}[t!]\centering
		\subfigure[][{\small $H_{\hat{b},\frac{\pi}{4}}$}]{
			\scalebox{0.25}[0.25]{\includegraphics{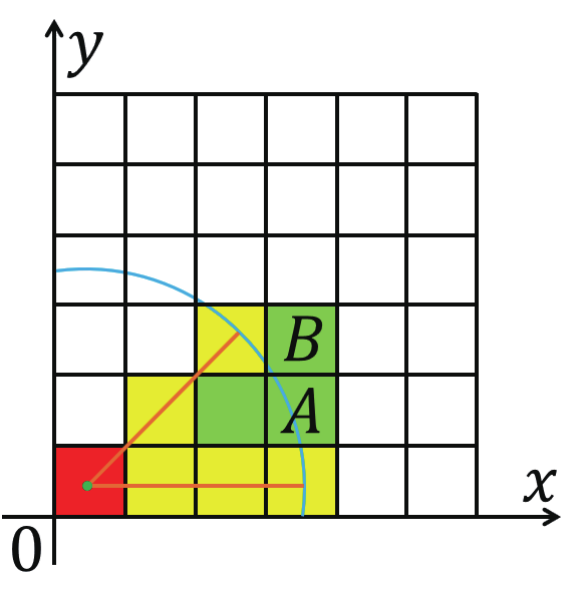}}
			\label{subfig:Appendix_so_count_A}}\hfill
		\subfigure[][{\small $H_{\hat{b},\frac{\pi}{4}} - 1$}]{
			\scalebox{0.25}[0.25]{\includegraphics{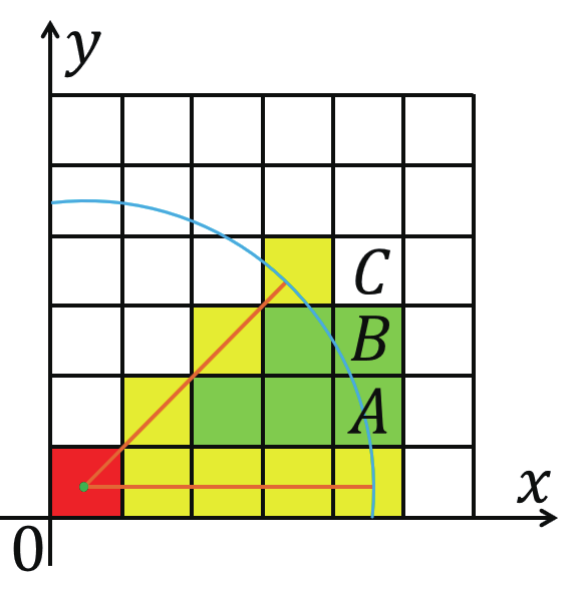}}
			\label{subfig:Appendix_so_count_B}}\hfill
		\subfigure[][{\small  The relation of $r$ and $r_1$}]{
			\scalebox{0.25}[0.25]{\includegraphics{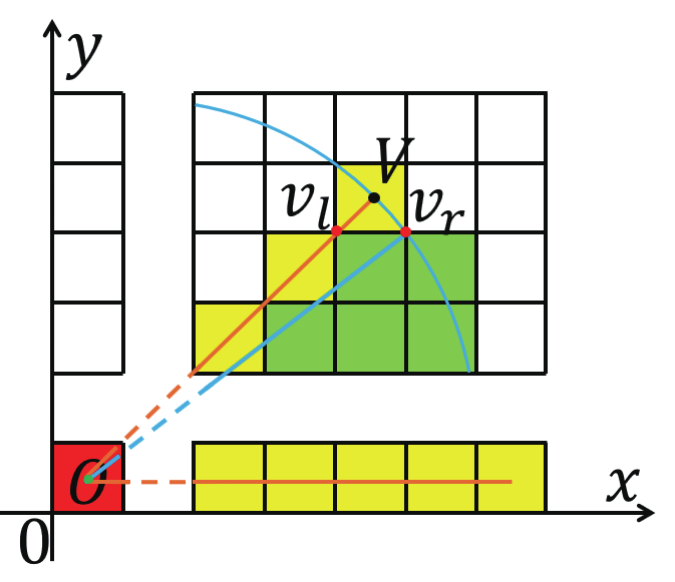}}
			\label{fig:Appendix_so_count2}}
		\caption{\small Different cross conditions and the relation.}
		\label{fig:Appendix_so_count}
	\end{figure}
	\begin{proof}
		Notice that given any an positive integer $\hat{b}$, each horizontal line in $S_{\hat{b},(0,\frac{\pi}{4})}$ contains and only contains one cell in $S_{\hat{b}}^O$.
		Let $H_{\hat{b},\frac{\pi}{4}}$ be the height of $S_{\hat{b},\frac{\pi}{4}}$ (the maximal $\hat{y}$ index in $S_{\hat{b},\frac{\pi}{4}}$).
		Then we have $H_{\hat{b},\frac{\pi}{4}}=\lceil\frac{\hat{b}}{\sqrt{2}}-\frac{1}{2}\rceil$.
		Let $H_{\hat{b},(0,\frac{\pi}{4})}$ be the height of $S_{\hat{b},(0,\frac{\pi}{4})}$.
		Then as is shown in Figure~\ref{fig:Appendix_so_count}, there are two kinds of relation between $H_{\hat{b},(0,\frac{\pi}{4})}$ and $H_{\hat{b},\frac{\pi}{4}}$.
		They are $H_{\hat{b},(0,\frac{\pi}{4})} = H_{\hat{b},\frac{\pi}{4}}$ (shown in Figure~\ref{subfig:Appendix_so_count_A}) and $H_{\hat{b},(0,\frac{\pi}{4})} = H_{\hat{b},\frac{\pi}{4}} - 1$ (shown in Figure~\ref{subfig:Appendix_so_count_B}).
		Let $V$ be the cell in $S_{\hat{b},(0,\frac{\pi}{4})}$ with largest height.
		In the condition of $H_{\hat{b},(0,\frac{\pi}{4})} = H_{\hat{b},\frac{\pi}{4}}$, the high probability circle border (blue circle) intersects the right boundary of $V$.
		In the condition of $H_{\hat{b},(0,\frac{\pi}{4})} = H_{\hat{b},\frac{\pi}{4}} - 1$, the high probability circle border intersects the lower boundary of $V$.

		Let $v_l$ and $v_r$ be the left bottom point and right bottom point of cell $V$ respectively.
		As is shown in Figure~\ref{fig:Appendix_so_count2},
		Let $r=Ov_r$, $r_1=Ov_l$, then $r=\hat{b}$, $r_1=\lfloor\frac{\hat{b}}{\sqrt{2}}-\frac{1}{2}\rfloor\cdot\sqrt{2}+\frac{1}{\sqrt{2}}$.
		In triangle $Ov_lv_r$, according to the law of cosines, $r=\sqrt{r_1^2+1^2-2r_1\cos{\frac{3\pi}{4}}}=\sqrt{r_1^2+1+\sqrt{2}r_1}$.
		When $\hat{b}<r$, the high probability circle border will intersect the bottom boundary of $V$, and $H_{\hat{b},(0,\frac{\pi}{4})} = H_{\hat{b},\frac{\pi}{4}} - 1$.
		When $\hat{b}>r$, the high probability circle border will intersect the right boundary of $V$, and $H_{\hat{b},(0,\frac{\pi}{4})} = H_{\hat{b},\frac{\pi}{4}}$.
		Notice that $\frac{r}{\hat{b}}\in [\frac{1}{\sqrt{2+\sqrt{2}}},1)\cup(1,\sqrt{2+\sqrt{2}}]\subseteq(\frac{1}{2},2)$.
		Thus, we have $H_{\hat{b},\frac{\pi}{4}}=H_{\hat{b},(0,\frac{\pi}{4})}+\lfloor\frac{r}{\hat{b}}\rfloor$.
		Therefore, we have $|S_{\hat{b}}^O|=H_{\hat{b},(0,\frac{\pi}{4})}=H_{\hat{b},\frac{\pi}{4}}-\lfloor\frac{r}{\hat{b}}\rfloor=\lceil\frac{\hat{b}}{\sqrt{2}}-\frac{1}{2}\rceil-\lfloor\frac{r}{\hat{b}}\rfloor$.
	\end{proof}
	
	\subsubsection{Proof of Theorem~\ref{thrm:SICount}}\label{appendix:thrm_SICount}
	\begin{proof}
		According to the definition of $S_{\hat{b},(0,\frac{\pi}{4})}, S_{\hat{b}}^I$ and $S_{\hat{b}}^O$, we have $|S_{\hat{b}}^I| = |S_{\hat{b},(0,\frac{\pi}{4})}| - |S_{\hat{b}}^O|$.
		Thus, we need to calculate $|S_{\hat{b},(0,\frac{\pi}{4})}|$.
		
		As is shown in Figure~\ref{subfig:Appendix_so_count3}, we can divide $S_{\hat{b},(0,\frac{\pi}{4})}$ into two parts: $S_{\hat{b},(0,\frac{\pi}{4})}^{(1)}$ (the green shadow part) and $S_{\hat{b},(0,\frac{\pi}{4})}^{(2)}$ (the red shadow part).
		Let $H_{\hat{b},(0,\frac{\pi}{4})}^L$ be the hight of $S_{\hat{b},(0,\frac{\pi}{4})}^{(1)}$.
		Then, $H_{\hat{b},(0,\frac{\pi}{4})}^L = \lceil\frac{\hat{b}}{\sqrt{2}}-\frac{1}{2}\rceil-1$,
		and $S_{\hat{b},(0,\frac{\pi}{4})}^{(1)}=\frac{H_{\hat{b},(0,\frac{\pi}{4})}^L(H_{\hat{b},(0,\frac{\pi}{4})}^L+1)}{2}$.
		Let $y_1,y_2,...,y_{|S_{\hat{b}}^{O}|}$ be the bottom borders of cells in $S_{\hat{b}}^{O}$.
		Then we have $y_i=i-\frac{1}{2}$ for $i=1,2,...|S_{\hat{b}}^{O}|$.
		The high probability circle intersects these bottom borders at points $(x_1,y_1),...(x_{|S_{\hat{b}}^{O}|}, y_{|S_{\hat{b}}^{O}|})$.
		We can calculate $x_i=\sqrt{\hat{b}^2-y_i^2}\in[\lceil\frac{\hat{b}}{\sqrt{2}}-\frac{1}{2}\rceil+\frac{1}{2},b)$ for $i=1,2,...|S_{\hat{b}}^{O}|$.
		Let the indexes of cells in $S_{\hat{b}}^{O}$ be $(x_1^V, y_1^V), (x_2^V, y_2^V),..., (x_{|S_{\hat{b}}^{O}|}^V, y_{|S_{\hat{b}}^{O}|}^V)$.
		Then we have $x_i^V = \lceil x_i-\frac{1}{2}\rceil$ and $y_i^V = i$ for $i=1,2,...|S_{\hat{b}}^{O}|$.
		Thus, $|S_{\hat{b},(0,\frac{\pi}{4})}^{(2)}|=\sum_{i=1}^{|S_{\hat{b}}^{O}|}(x_i^V-H_{\hat{b},(0,\frac{\pi}{4})}^L)$.
		Therefore, we have
		\begin{equation}
			{\scriptsize
				\begin{aligned}
					|S_{\hat{b},(0,\frac{\pi}{4})}| &= |S_{\hat{b},(0,\frac{\pi}{4})}^{(1)}|+|S_{\hat{b},(0,\frac{\pi}{4})}^{(2)}|\\
					&= \frac{H_{\hat{b},(0,\frac{\pi}{4})}^L(H_{\hat{b},(0,\frac{\pi}{4})}^L+1)}{2} + \sum_{i=1}^{|S_{\hat{b}}^{O}|}(x_i^V-H_{\hat{b},(0,\frac{\pi}{4})}^L)\\
					&= \frac{(\lceil\frac{\hat{b}}{\sqrt{2}}-\frac{1}{2}\rceil-1)\lceil\frac{\hat{b}}{\sqrt{2}}-\frac{1}{2}\rceil}{2}
					- |S_{\hat{b}}^{O}|(\lceil\frac{\hat{b}}{\sqrt{2}}-\frac{1}{2}\rceil-1)\\
					&+ \sum_{i=1}^{|S_{\hat{b}}^{O}|}\lceil \sqrt{\hat{b}^2-(i-\frac{1}{2})^2}-\frac{1}{2}\rceil\\
					&= \frac{1}{2}(\lceil\frac{\hat{b}}{\sqrt{2}}-\frac{1}{2}\rceil-1)(\lceil\frac{\hat{b}}{\sqrt{2}}-\frac{1}{2}\rceil-2|S_{\hat{b}}^{O}|)\\
					&+ \sum_{i=1}^{|S_{\hat{b}}^{O}|}\lceil\sqrt{\hat{b}^2-(i-\frac{1}{2})^2}-\frac{1}{2}\rceil,
				\end{aligned}
			}
		\end{equation}
		and
		\begin{equation}
			{\scriptsize
				\begin{aligned}
					|S_{\hat{b}}^I| &= |S_{\hat{b},(0,\frac{\pi}{4})}|-|S_{\hat{b}}^O|\\
					&= \frac{1}{2}(\lceil\frac{\hat{b}}{\sqrt{2}}-\frac{1}{2}\rceil-1)(\lceil\frac{\hat{b}}{\sqrt{2}}-\frac{1}{2}\rceil-2|S_{\hat{b}}^{O}|)\\
					&+ \sum_{i=1}^{|S_{\hat{b}}^{O}|}\lceil\sqrt{\hat{b}^2-(i-\frac{1}{2})^2}-\frac{1}{2}\rceil -|S_{\hat{b}}^{O}|\\
					&= \frac{1}{2}\lceil\frac{\hat{b}}{\sqrt{2}}-\frac{1}{2}\rceil(\lceil\frac{\hat{b}}{\sqrt{2}}-\frac{1}{2}\rceil-2|S_{\hat{b}}^{O}|-1)\\
					&+ \sum_{i=1}^{|S_{\hat{b}}^{O}|}\lceil\sqrt{\hat{b}^2-(i-\frac{1}{2})^2}-\frac{1}{2}\rceil.
				\end{aligned}
			}
		\end{equation}
	\end{proof}
	
	\subsection{Experiment on Crime with Full domain}\label{appendix:Crime_full}
	In order to makes the experiment more comprehensive, we compare all these methods on the Crime data set with full domain.
	The comparison result is shown in Figure~\ref{fig:alter_d_e_add}.
	We observe that, the difference in Wasserstein distance among all theses methods on Crime with the full domain is similar to that with the coarse domain. 
	Our \solutionB{} outperforms other methods when $d$ is large.
	The only different is that when the privacy budget $\epsilon$ increases, \solutionCMPg{} slightly outperforms our \solutionB{}.
	That is because LDP methods in coarse data have fewer non-zeros entries, resulting in low information content in the original data. Consequently, the introduction of noise can obscure the original signal, making it difficult to extract accurate information. 
	
	\begin{figure}[t!]\centering
		\subfigure{
			\scalebox{0.35}[0.35]{\includegraphics{figures/result_bar/bar_other_change_trim_new.pdf}}}\hfill\\
		\addtocounter{subfigure}{-1}
		\subfigure[][{\small $W_2$ with small $d$ varied}]{
			\scalebox{0.24}[0.24]{\includegraphics{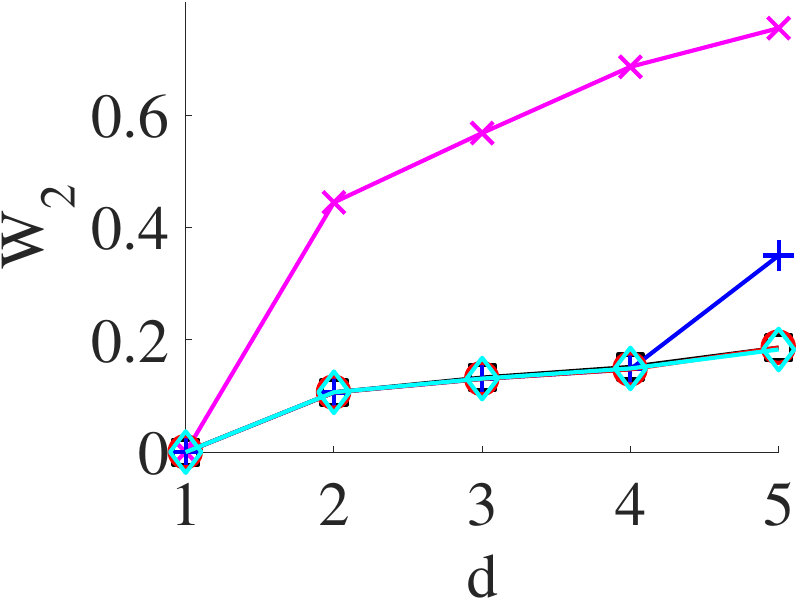}}
			\label{subfig:alter_d_crime_add}}\hspace{1cm}	
		\subfigure[][{\small $W_2$ with large $d$ varied}]{
			\scalebox{0.24}[0.24]{\includegraphics{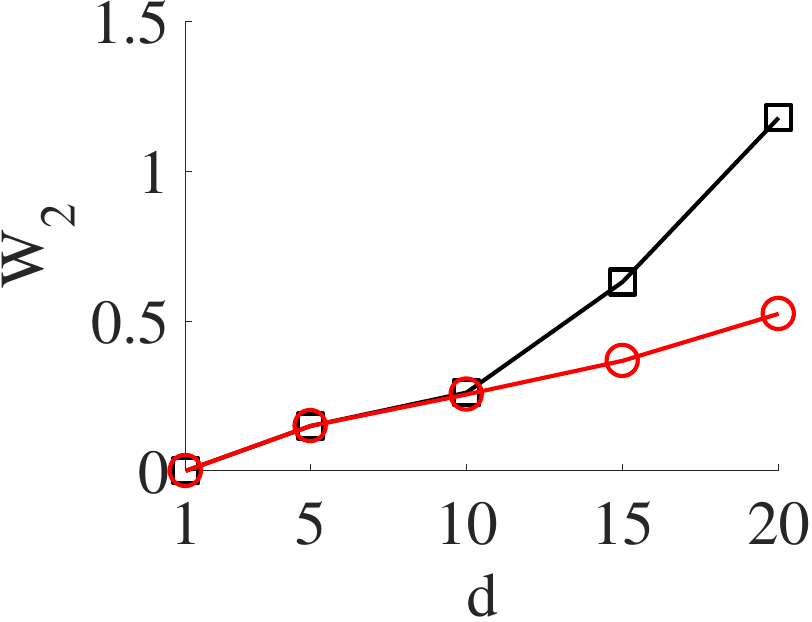}}
			\label{subfig:alter_g_crime_extended_add}}\\ 	
		\subfigure[][{\small $W_2$ with small $\epsilon$ varied}]{
			\scalebox{0.24}[0.24]{\includegraphics{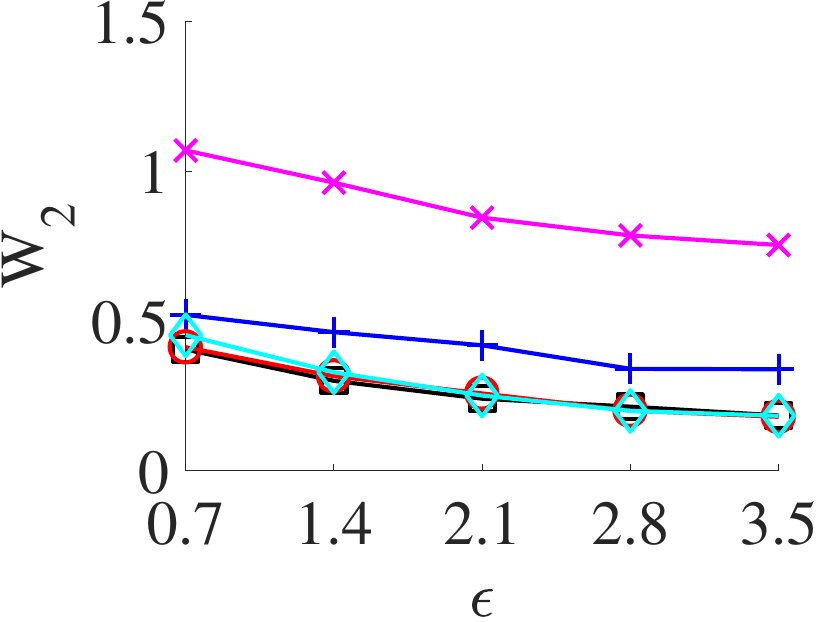}}
			\label{subfig:alter_e_crime_add}}\hspace{1cm} 
		\subfigure[][{\small $W_2$ with large $\epsilon$ varied}]{
			\scalebox{0.24}[0.24]{\includegraphics{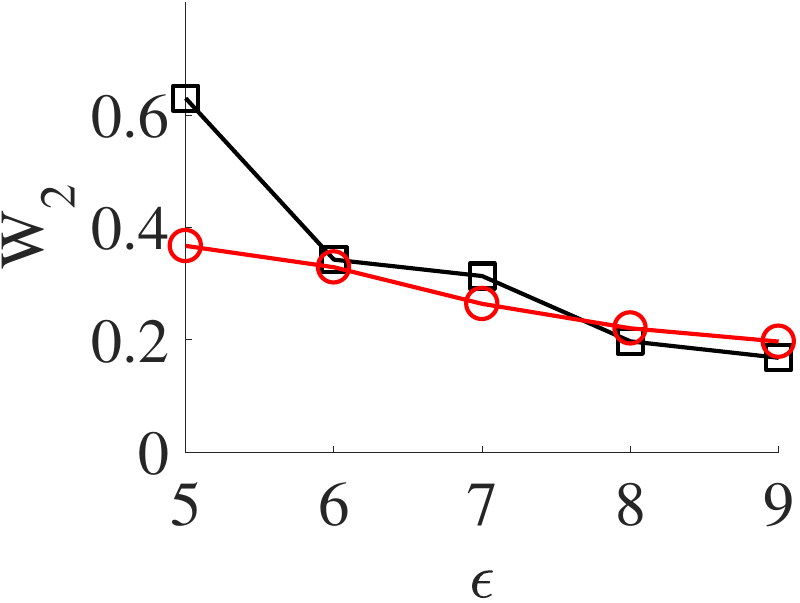}}
			\label{subfig:alter_e_crime_extended_add}}\hfill 
		\caption{\small Wasserstein distances with $d$ or $\epsilon$ varied on Crime with total domain.}
		\label{fig:alter_d_e_add}
	\end{figure}
	
	\subsection{Experiment for Trajectory on NYC}\label{appendix:TrajectoryCmp}
	We compared our \solutionB{} with \solutionCMPTrATotalName{} (\solutionCMPTrA{})~\cite{DBLP:journals/pvldb/DuHZFCZG23}
	and \solutionCMPTrBTotalName{} (\solutionCMPTrB{})~\cite{DBLP:journals/pvldb/Zhang000H23}.
	
	\noindent
	\textbf{Data sets.} We use \emph{NYC Green Taxis}~\cite{NYCGreen2016} (NYC) for our trajectory experiment.
	To generate original trajectory data, we first divide the data set domain into a $300\times 300$ grid.
	Then we map each grid cell to the points within it.
	Next, we randomly select $1,000$ trajectory start cells and $1,000$ trajectory lengths ranging from $2$ to $200$.
	For each trajectory start cell, we iteratively choose one of its neighbors with probability proportional to the point quantity within these neighboring cells until reaching the sampled trajectory length. 
	For each chosen cell, we randomly select one point within it.
	Ultimately, we obtain $1,000$ point lists, each list containing $2$ to $200$ points.
	We refer to each of these lists as a sampled trajectory.
	
	\noindent
	\textbf{Mechanisms.}
	\solutionCMPTrA{} and \solutionCMPTrB{} primarily focus on trajectory estimation.
	To make these methods comparable with \solutionB{}, we transform trajectory statistics to point statistics.
	We follow the following seven steps.
	\begin{enumerate}
		\item[(1)] Divide the trajectory input domain as $d\times d$ grids.
		\item[(2)] Count the original trajectory points in each grid cell.
		\item[(3)] Approximate the real distribution $D_T$ using the counting results.
		\item[(4)] Apply \solutionCMPTrA{} or \solutionCMPTrB{} to estimate trajectories.
		\item[(5)] Count the estimated trajectory points in each grid cell.
		\item[(6)] Approximate the estimated distribution $\hat{D}_T$ using counting results.
		\item[(7)] Calculate the Wasserstein distance $W_2$ between $D_T$ and $\hat{D}_T$.
	\end{enumerate}
	
	\noindent
	\textbf{Parameter Settings.}
	We vary the discrete side length $d$ from $1$ to $20$, with default value of $15$.
	Additionally, we adjust the privacy budget $\epsilon$ from $0.5$ to $2.5$, setting the default value as $1.5$. Table~\ref{tab:settings2} presents these parameter settings.
	
	\begin{table}[t!]
		\caption{\small Experimental Settings for Trajectory Cases.} \label{tab:settings2}
		\centering
		\resizebox{6.8cm}{!}{
			\begin{tabular}{l|l}
				{\bf \quad \quad Parameters} & {\bf \qquad \qquad Values} \\ \hline \hline
				the discrete side length, $d$                   &  $1$, $5$, $10$, $\mathbf{15}$, $20$\\\hline
				the privacy budget, $\epsilon$      & $0.5$, $1.0$, $\textbf{1.5}$, $2.0$, $2.5$\\
				\hline
			\end{tabular}
		}\vspace{-1ex}
	\end{table}
	
	\noindent
	\textbf{Experimental Results.}
	We continue use 2-norm Wasserstein distance $W_2$ as the error metric in this experiment.
	Figure~\ref{fig:alter_d_e_trajectory} presents the experimental results.
	As the grid side length $d$ increases, the Wasserstain distance for all three mechanisms increase.
	This occurs because a larger $d$ results in more grid cells, widening the gap between the recovered and actual density distributions.
	Our \solutionB{} consistently outperforms the other two mechanisms.
	\solutionCMPTrA{} and \solutionCMPTrB{} allocate more of their privacy budget to determining trajectory direction, leading to suboptimal performance.
	Furthermore, as the privacy budget increases, \solutionCMPTrA{} fluctuates, while \solutionCMPTrB and \solutionB{} decrease (\solutionCMPTrB{} decreases slightly and less noticeable). This difference likely stems from \solutionCMPTrA{}'s reliance on a constructed synthetic data set, which captures trajectory information but fail to fully represent the underlying spatial data distribution. Our \solutionB{} is always better than \solutionCMPTrA{} and \solutionCMPTrB{}. 
	
	\begin{figure}[t!]\centering
		\subfigure{
			\scalebox{0.35}[0.35]{\includegraphics{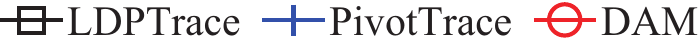}}}\hfill\\
		\addtocounter{subfigure}{-1}
		\subfigure[][{\small $W_2$ with $d$ varied in trajectory cases}]{
			\scalebox{0.24}[0.24]{\includegraphics{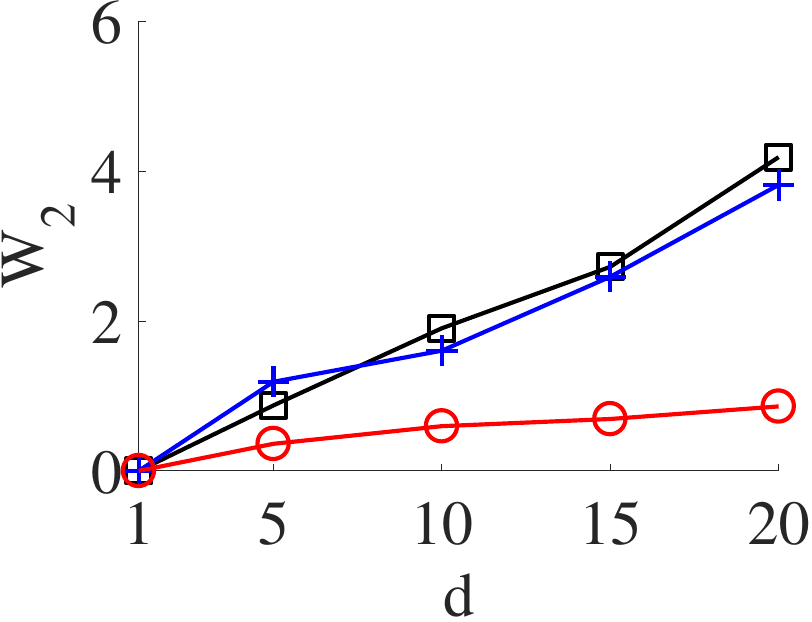}}
			\label{subfig:alter_g_nyc_extended_trajectory}}
		\hspace{1cm}
		\subfigure[][{\small $W_2$ with $\epsilon$ varied in trajectory cases}]{
			\scalebox{0.24}[0.24]{\includegraphics{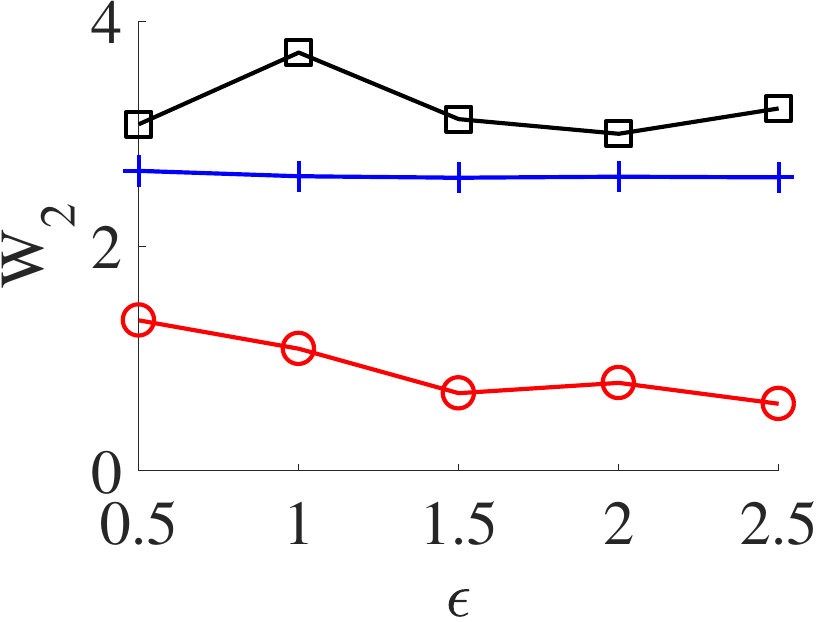}}
			\label{subfig:alter_e_nyc_trajectory}}\hfill
		\caption{\small Wasserstein distances with $d$ or $\epsilon$ varied on NYC in trajectory cases.}
		\label{fig:alter_d_e_trajectory}
	\end{figure}

\end{document}